\documentclass[letterpaper, 10 pt, conference]{ieeeconf}%
\usepackage{graphicx}
\usepackage{epsfig}
\usepackage{mathptmx}
\usepackage{times}
\usepackage{amsmath}
\usepackage{thmtools,thm-restate}
\usepackage{amssymb}
\usepackage[section]{placeins}
\usepackage{placeins}
\usepackage{amsmath}
\usepackage{multirow}
\usepackage{graphicx}
\usepackage[normalem]{ulem}
\usepackage{hyperref}
\usepackage{subcaption}
\usepackage{algorithm,tabularx}
\usepackage{algpseudocode}
\usepackage{amsfonts}
\usepackage{cases}
\usepackage{romannum}
\usepackage{textcomp}
\usepackage{float}
\usepackage{xcolor}%
\providecommand{\U}[1]{\protect\rule{.1in}{.1in}}
\providecommand{\U}[1]{\protect\rule{.1in}{.1in}}
\IEEEoverridecommandlockouts
\newcommand{\R}{\mathbb{R}}

\newtheorem{assumption}{Assumption}
\newtheorem{theorem}{Theorem}
\newtheorem{corollary}{Corollary}
\newtheorem{lemma}{Lemma}
\newtheorem{proposition}{Proposition}
\newtheorem{remark}{Remark}
\useunder{\uline}{\ul}{}
\makeatletter
\newcommand{\multiline}[1]{  \begin{tabularx}{\dimexpr\linewidth-\ALG@thistlm}[t]{@{}X@{}}
#1
\end{tabularx}
}
\makeatother

\usepackage{tikz}

\begin{document}

\title{{\LARGE \textbf{Event-Driven Receding Horizon Control of Energy-Aware Dynamic Agents For Distributed Persistent Monitoring}}\vspace{-10pt}}
\author{Shirantha Welikala and Christos G. Cassandras \thanks{$^{\star}$Supported in part by NSF under grants ECCS-1931600, DMS-1664644, CNS-1645681, by AFOSR under grant FA9550-19-1-0158, by ARPA-E under grant DE-AR0001282 and by the NEXTCAR program under grant DE-AR0000796 and by the MathWorks.} \thanks{The authors are with the Division of Systems Engineering and Center for Information and Systems Engineering, Boston University, Brookline, MA 02446, \texttt{{\small \{shiran27,cgc\}@bu.edu}}.}}
\maketitle

\begin{abstract}
This paper addresses the persistent monitoring problem defined on a network where a set of nodes (targets) needs to be monitored by a team of dynamic energy-aware agents. The objective is to control the agents' motion to jointly optimize the overall agent energy consumption and a measure of overall node state uncertainty, evaluated over a finite period of interest. To achieve these objectives, we extend an established event-driven Receding Horizon Control (RHC) solution by adding an optimal controller to account for agent motion dynamics and associated energy consumption. The resulting RHC solution is computationally efficient, distributed and on-line. Finally, numerical results are provided highlighting improvements compared to an existing RHC solution that uses energy-agnostic first-order agents. 
\end{abstract}

\thispagestyle{empty} \pagestyle{empty}

\section{Introduction}


We consider the problem of controlling a group of mobile \emph{agents} deployed to monitor a finite set of ``points of interest'' (henceforth called \emph{targets}) in a mission space. In particular, each agent follows second-order unicycle dynamics and each target has an ``uncertainty'' metric associated with its state that increases when no agent is monitoring (i.e., sensing or collecting information from) the target and decreases when one or more agents are monitoring it by dwelling in its vicinity. The goal is to optimally control each agent's motion so as to collectively minimize the overall agent energy consumption and a measure of target uncertainties - evaluated over a fixed period of interest. This problem setup is widely known as the \emph{persistent monitoring} problem and it encompasses applications such as  
environmental sensing \cite{Elwin2020}, 
surveillance \cite{Kingston2008}, 
traffic monitoring \cite{Reshma2016},
data collection \cite{Smith2011},
event detection \cite{Yu2015} and 
energy management \cite{Mathew2015}. In order to suit different application scenarios, this persistent monitoring problem has been studied in the literature under different objective functions \cite{Hari2019}, agent dynamic models \cite{Welikala2020J4,Wang2017} and target state dynamic models \cite{Zhou2019,Lan2013}.

A common way to categorize persistent monitoring problem setups is based on whether the shapes of trajectory segments (available for the agents to travel between targets) are predefined \cite{Smith2011,Zhou2019} or not \cite{Lan2013,Lin2013}. In the latter case, the main challenge is to search for the optimal agent trajectory shapes. This is often achieved by restricting agent trajectory shapes to specific parametric families (elliptical, Fourier, etc. \cite{Lin2013}) and optimizing the objective function of interest within these families. In contrast, when the shapes of trajectory segments are predefined, the challenge is to search for: 1) the optimal target visiting schedules of agents and 2) the optimal control laws to govern agents on corresponding trajectory segments - assuming an agent has to remain stationary on a target to monitor it. As introduced in \cite{Zhou2019} and illustrated in Fig. \ref{Fig:GraphAbstraction}, this can be seen as a Persistent Monitoring on a Network (PMN) problem where targets and trajectory segments are modeled as nodes and edges of a network, respectively. Such PMN problems are significantly more complicated than the NP-hard traveling salesman problems \cite{Kirk2020} and thus have inspired many different solution approaches \cite{Welikala2020J4,Zhou2019,Rezazadeh2019}. 

\begin{figure}[!h]
    \centering
    \includegraphics[width=3in]{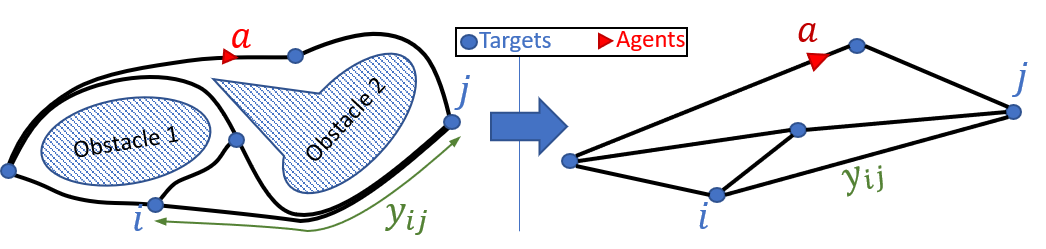}
    \caption{The network abstraction.}
    \label{Fig:GraphAbstraction}
    \vspace{-2mm}
\end{figure}

The work in \cite{Rezazadeh2019} proposes a centralized off-line greedy algorithm to determine the optimal target visiting schedules of agents (i.e., each agent's sequence of targets to visit and respective dwell-times to be spent at visited targets) in PMN problems. In contrast, for the same task, \cite{Zhou2019} proposes a gradient-based distributed on-line approach - which, however, requires a brief centralized off-line initialization stage to address non-convexities. An alternative approach is taken in the recent work \cite{Welikala2020J4} which exploits the event-driven nature of PMN systems to develop a distributed on-line solution based on event-driven Receding Horizon Control (RHC) \cite{Li2006}. 
This RHC solution enjoys many promising features such as being computationally cheap, parameter-free, gradient-free and robust in the presence of various forms of state and system perturbations. 

However, the work mentioned above \cite{Welikala2020J4,Zhou2019,Rezazadeh2019} ignores agent dynamics by assuming each trajectory segment has a predefined transit-time value that an agent has to spend in order to travel on it. This assumption allows one to focus on determining the optimal target visiting schedules of agents, ignoring how the agents are governed during the transition periods where they travel on trajectory segments. In essence, it is identical to assuming each agent follows a first-order dynamic model controlled by its velocity.

In contrast, in this paper, we assume each agent follows a second-order dynamic model governed by acceleration rather than velocity. This leads to a better approximation of actual agent behaviors in practice and smoother agent state trajectories \cite{Wang2017}. In particular, we incorporate agent energy consumption into the objective function to limit agent accelerations and velocities and also to motivate agents to make energy-efficient decisions. Under these modifications, we show how each agent needs to optimally select each transit-time value on its trajectory based on current local state information - instead of using a fixed set of predefined transit-time values. In particular, we explicitly derive optimal control laws to govern each agent on each trajectory segment. Finally, we not only compare the improvements achieved with respect to an existing RHC solution \cite{Welikala2020J4} that uses energy-agnostic first-order agents but also derive energy-aware optimal control laws for even such first-order agents. 

In this paper, first, we show that each agent's trajectory is fully characterized by the sequence of decisions it makes at specific discrete event times in its trajectory. Second, considering an agent at each such event-time, we formulate a Receding Horizon Control Problem (RHCP) that determines the agent's optimal immediate control decisions over an \emph{optimally determined} planning horizon. 
These control decisions are subsequently executed over a shorter action horizon defined by the next event that the agent observes, and the same process is continued in this event-driven manner. 
As the third step, we show that this RHCP includes an optimal control component and it is then solved considering energy-aware second-order agents.
Finally, several different numerical examples (i.e.,  PMN problems) are used to compare the developed RHC solution with respect to the RHC solution proposed in \cite{Welikala2020J4} that uses energy-agnostic first-order agents.

This paper is organized as follows. Section \ref{Sec:ProblemFormulation} presents the problem formulation and overview of the RHC approach. Sections \ref{Sec:SolutionToTheRHCPs} and \ref{Sec:FOMethods} present the formulation and solution of the RHCP with second-order agents and first-order agents, respectively. Numerical results are provided in Section \ref{Sec:ExtentionsAndExamples}. Finally, Section \ref{Sec:Conclusion} concludes the paper.


\section{Problem Formulation}
\label{Sec:ProblemFormulation}
We consider a $2$-dimensional mission space containing $M$ targets (nodes) in the set $\mathcal{T}=\{1,2,\ldots,M\}$ where the location of target $i\in\mathcal{T}$ is fixed at $Y_{i}\in\mathbb{R}^{2}$. A team of $N$ agents in the set $\mathcal{A}=\{1,2,\ldots,N\}$ is deployed to monitor the targets. Each agent $a\in\mathcal{A}$ moves within this mission space where its location and orientation at time $t$ are denoted by $s_{a}(t)\in\mathbb{R}^{2}$ and $\theta_a(t)\in[0,2\pi]$, respectively.

\paragraph{\textbf{Target Model}}

Each target $i\in\mathcal{T}$ has an associated \emph{uncertainty state} $R_{i}(t)\in\mathbb{R}$ which follows the dynamics \cite{Zhou2019}:
\begin{equation}
\dot{R}_{i}(t)=
\begin{cases}
A_{i}-B_{i}N_{i}(t) &\mbox{ if } R_{i}(t)>0\mbox{ or } A_{i}-B_{i}N_{i}(t)>0\\
0 & \mbox{ otherwise,}
\end{cases}
\label{Eq:TargetDynamics}
\end{equation}
where $N_{i}(t)=\sum_{a\in\mathcal{A}}\mathbf{1}\{s_{a}(t)=Y_{i}\}$($\mathbf{1}\{\cdot\}$ denotes the indicator function) is the number of agents present at target $i$ at time $t$. According to \eqref{Eq:TargetDynamics}: (i) $R_i(t)$ increases at a rate $A_i$ when no agent is visiting target $i$, (ii) $R_i(t)$ decreases at a rate $A_i-B_iN_i(t)$ where $B_i$ is the uncertainty removal rate by a visiting agent to the target $i$ and (iii) $R_i(t)\geq0, \,\forall t$.

\paragraph{\textbf{Agent Model}}
The location and orientation $(s_a(t),\,\theta_a(t))$ of an agent $a\in\mathcal{A}$ follows the second-order unicycle dynamics given by
\begin{equation}\label{Eq:AgentDynamics}
\begin{aligned} 
\dot{s}_a(t) 
&= v_a(t)
\begin{bmatrix}
\cos(\theta_a(t)) &
\sin(\theta_a(t))
\end{bmatrix}^T,\\
\dot{v}_a(t) &= u_a(t),\\
\dot{\theta}_a(t) &= w_a(t), 
\end{aligned}
\end{equation}
where $v_a(t)$ is the tangential velocity, $u_a(t)$ is the tangential acceleration and $w_a(t)$ is the angular velocity. We consider $u_a(t)$ and $w_a(t)$ as the agent control inputs.

Note that according to \eqref{Eq:TargetDynamics}, the agent has to stay stationary on a target  $i\in\mathcal{T}$ for some positive amount of time to contribute to decreasing a positive target uncertainty $R_i(t)$. Therefore, during such a \emph{dwell-time} period, the agent must enforce $u_a(t)=v_a(t)=0$ with $s_a(t) = Y_i$. 


\paragraph{\textbf{Objective}}
Our aim is to minimize the composite objective $J_T$ of the \emph{total energy spent} $J_e$ (called the \emph{energy objective}) and the \emph{mean system uncertainty} $J_s$ (called the \emph{sensing objective}) over a finite time interval $[0,T]$:
\begin{equation} 
J_{T} 
\triangleq \alpha J_e + J_s 
= \alpha 
\underbrace{\int_0^T \sum_{a\in\mathcal{A}}u_a^2(t)\,dt}_{\triangleq\ J_e} 
+ 
\underbrace{\frac{1}{T}\int_{0}^{T}\sum_{i\in\mathcal{T}}R_{i}(t)\,dt}_{\triangleq\ J_s},
\label{Eq:MainObjective}%
\end{equation}
by controlling agent control inputs $u_a(t),w_a(t),\forall a\in\mathcal{A},t\in[0,T]$. Note that $\alpha$ in \eqref{Eq:MainObjective} is a weight factor that can also be manipulated to constrain the resulting optimal agent controls (details on selecting $\alpha$ to ensure proper normalization of the $J_T$ components are provided in Appendix \ref{App:NormalizationFactor}). Note also that the cost of angular velocity (steering) control is not included in \eqref{Eq:MainObjective}. The trade-off between $J_e$ and $J_s$ components of \eqref{Eq:MainObjective} is clear from the fact that the aggressiveness of agent transitions in-between targets affects negatively the $J_e$ component but positively the $J_s$ component.

\paragraph{\textbf{Graph Topology}}
We embed a directed graph topology $\mathcal{G}=(\mathcal{T},\mathcal{E})$ into the mission space so that the \emph{targets} are represented by the graph \emph{vertices} $\mathcal{T}=\{1,2,\ldots,M\}$ and the inter-target \emph{trajectory segments} are represented by the graph \emph{edges} $\mathcal{E}\subseteq\{(i,j):i,j\in\mathcal{T}\}$ (see also Fig. \ref{Fig:GraphAbstraction}). These trajectory segments may take arbitrary (prespecified) shapes so as to account for constraints in the mission space and agent motion. We use $\rho_{ij}$ to denote the \emph{transit-time} that an agent spends on a trajectory segment $(i,j)\in\mathcal{E}$ to reach target $j$ from target $i$. In contrast to \cite{Zhou2019} and \cite{Welikala2020J4} where these transit-time values were treated as predefined, in this work they are considered as control-dependent.
We also use $\mathcal{P}_{ij}$ to represent the \emph{transit-time interval} ($\mathcal{P}_{ij}\subset[0,T]$ of length $\rho_{ij}$) corresponding to the transit-time $\rho_{ij}$.

The \emph{neighbor set} and the \emph{neighborhood} of a target $i\in\mathcal{T}$ are defined based on the available trajectory segments $\mathcal{E}$ as
\begin{equation}
\mathcal{N}_{i}\triangleq\{j:(i,j)\in\mathcal{E}\}\mbox{ and }\bar{\mathcal{N}}_{i}=\mathcal{N}_{i}\cup\{i\}.\label{Eq:Neighborset}%
\end{equation}


\paragraph{\textbf{Control}}
As stated earlier, when an agent $a\in\mathcal{A}$ dwells on a target $i\in\mathcal{T}$, the agent control $u_a(t)$ is zero. However, over such a dwell-time period, the agent control $w_a(t)$ may or may not be zero (exact details will be provided later).
Next, when the agent is ready to leave the target $i$, it needs to decide the \emph{next-visit} target $j\in\mathcal{N}_{i}$ along with the corresponding control profiles $u_a(t),w_a(t)$ to be used on the trajectory segment $(i,j)\in\mathcal{E}$ over $t\in\mathcal{P}_{ij}$.

In essence, the overall control exerted on an agent can be seen as a sequence of: \emph{dwell-times} $\delta_{i}\in\mathbb{R}_{\geq0}$, \emph{next-visit} targets $j\in\mathcal{N}_{i}$ and \emph{control profile segments} $\{(u_a(\tau),w_a(\tau)): \tau\in\mathcal{P}_{ij}\}$. Our goal is to determine $(\delta_{i}(t_s),j(t_s),\{(u_a(\tau),w_a(\tau)): \tau\in\mathcal{P}_{ij}(t_s)\})$ for any agent $a\in\mathcal{A}$ residing at any target $i\in\mathcal{T}$ at any time $t_s\in [0,T]$, which is optimal in the sense of minimizing \eqref{Eq:MainObjective}. 

Clearly, this PMN problem is more complicated than the well known NP-Hard \emph{traveling salesman problem} (TSP) \cite{Kirk2020} due to its inclusion of: (i) multiple agents, (ii) target dynamics, (iii) agent dynamics, (iv) target dwell-times and (v) repeated target visits. Even though one can still resort to dynamic programming techniques to solve this PMN problem, for all the above reasons, the problem is intractable - even for the most simplistic problem configurations. 

\paragraph{\textbf{Receding Horizon Control}}
As a solution to this PMN problem, inspired by the prior work \cite{Welikala2020J4} (where we dealt with first-order agents without agent energy concerns), this paper proposes an \emph{Event-Driven} \emph{Receding Horizon Controller} (RHC) at each agent. The key idea behind RHC derives from Model Predictive Control (MPC). However, RHC exploits the problem's event-driven nature to significantly reduce the complexity by effectively decreasing the frequency of control updates. As introduced and extended later on in \cite{Li2006} and \cite{Chen2020,Welikala2020J4} respectively, the RHC is invoked by the agents in a \emph{distributed} manner at specific events of interest in their trajectories. Upon invoking it, RHC determines the agent controls that optimize the objective \eqref{Eq:MainObjective} over a \emph{planning horizon} and subsequently executes the determined optimal controls over a shorter \emph{action horizon}.

In particular, when the RHC is invoked at some event-time $t_s\in[0,T]$ by an agent $a\in\mathcal{A}$ while residing at target $i\in\mathcal{T}$, it determines: (i) the remaining dwell-time $\delta_i(t_s)$ at target $i$, (ii) the next-visit target $j(t_s)\in\mathcal{N}_{i}$, (iii) the control profile segments $\{u_a(\tau),w_a(\tau): \tau\in\mathcal{P}_{ij}(t_s)\}$ and (iv) the dwell-time $\delta_j(t_s)$ at target $j(t_s)$. These control decisions are jointly represented by $U_{ia}(t_s)$ and its optimal value is determined by solving an optimization problem of the form:
\begin{equation}
\begin{alignedat}{4}
& U_{ia}^{\ast}(t_s) = & \underset{\makebox[1.5cm]{\footnotesize  $U_{ia}(t_s)\in\mathbb{U}(t_s)$}}{\arg\min} &  && 
  J_{H}(X_{ia}(t_s),U_{ia}(t_s);H)
  +\hat{J}_{H}(X_{ia}(t_s+H)) &
\end{alignedat}
 \label{RHC problem0}%
\end{equation}
where $X_{ia}(t_s)$ is the current local state and $\mathbb{U}(t_s)$ is the feasible control set at time $t_s$ (exact definitions are provided later). The term $J_{H}(X_{ia}(t_s),U_{ia}(t_s);H)$ represents the immediate cost over the planning horizon $[t_s,t_s+H]$ and $\hat{J}_{H}(X_{ia}(t_s+H)$ is an estimate of the future cost based on the state at $t_s+H$. 

In particular, we follow the \emph{variable horizon} concept proposed in \cite{Welikala2020J4} where the planning horizon length is treated as an upper-bounded function of control decisions $\mathsf{w}(U_{ia}(t_s)) \leq H$ rather than an exogenously selected value $H$, and the $\hat{J}_{H}(X_{ia}(t_s+H)$ term is ignored. Hence, this approach incorporates the selection of planning horizon length $\mathsf{w}(U_{ia}(t_s))$ into the optimization problem \eqref{RHC problem0}, which now can be re-stated as 
\begin{equation}\label{RHC problem}
\begin{alignedat}{4}
& U_{ia}^{\ast}(t_s) = & \underset{\makebox[2cm]{\footnotesize $U_{ia}(t_s)\in\mathbb{U}(t_s)$}}{\arg\min} & \ && J_{H}(X_{ia}(t_s),U_{ia}(t_s);\ \mathsf{w}(U_{ia}(t_s))) & \\
&  & \makebox[2cm]{subject to} &  && \mathsf{w}(U_{ia}(t_s))\leq H. &
\end{alignedat}
\end{equation}

\subsection{Preliminary Results}
According to \eqref{Eq:TargetDynamics}, the target state (uncertainty) $R_{i}(t)$ of a target $i\in\mathcal{T}$ is piece-wise linear and its gradient $\dot{R}_{i}(t)$ changes only when one of the following (strictly local) \emph{events} occurs:
(\romannum{1}) An agent arrival at $i$, 
(\romannum{2}) $R_i(t)$ switches from positive to zero, denoted as $[R_{i}(t)\rightarrow 0^{+}]$, or 
(\romannum{3}) An agent departure from $i$. 
Let us denote the sequence of such event times (associated with the target $i$) as $t_{i}^{k}$ where $k \in \mathbb{Z}_{>0}$ with $t_{i}^{0}=0$. Then, it is easy to see from \eqref{Eq:TargetDynamics} that 
\begin{equation}
\dot{R}_{i}(t)=\dot{R}_{i}(t_{i}^{k}),\ \forall t\in\lbrack t_{i}^{k}%
,t_{i}^{k+1}).\label{Eq:TargetDynamics2}%
\end{equation}

\begin{remark}
\label{Rm:NonOverlapping} 
As pointed out in \cite{Yu2016, Welikala2020J4} (and the references therein), allowing multiple agents to simultaneously reside on a target (known also as \textquotedblleft simultaneous target sharing\textquotedblright) is known to lead to solutions with poor performance levels. Thus, we enforce a constraint \cite{Welikala2020J4} on the controller to ensure:
\begin{equation}\label{Eq:NoTargetSharing}
    N_{i}(t)\in\{0,1\},\ \forall t\in [0,T],\ \forall i\in\mathcal{T}.
\end{equation}
Clearly, this constraint only applies if $N\geq2$. 
\end{remark}

Under \eqref{Eq:NoTargetSharing}, it follows from \eqref{Eq:TargetDynamics} and \eqref{Eq:TargetDynamics2} that the sequence $\{\dot{R}_{i}(t_{i}^{k})\}_{k=0,1,\ldots}$ is a \emph{cyclic order} of three elements: $\{-(B_{i}-A_{i}),0,A_{i}\}$. Next, in order to make sure that each agent is capable of enforcing the event $[R_{i}\rightarrow0^{+}]$ at any target $i\in\mathcal{T}$, we assume the following simple stability condition \cite{Welikala2020J4}:

\begin{assumption}
\label{As:TargetUncertaintyRateInequality} 
Target uncertainty rate parameters $A_{i}$ and $B_{i}$ of each target $i\in\mathcal{T}$ satisfy $0<A_{i}<B_{i}$.
\end{assumption}

\paragraph{\textbf{Decomposition of the Sensing Objective $J_s$}} 
The following theorem provides a target-wise and temporal decomposition of the sensing objective $J_s$ defined in \eqref{Eq:MainObjective}.
\begin{theorem}(\cite[Th.1]{Welikala2020J4})
\label{Th:ContributionTarget} 
The contribution to the term $J_{s}$ in \eqref{Eq:MainObjective} by a
target $i\in\mathcal{T}$ during a time period $[t_{0},t_{1})\subseteq\lbrack
t_{i}^{k},t_{i}^{k+1})$ for some $k\in\mathbb{Z}_{\geq0}$ is $\frac{1}{T}J_{i}(t_{0},t_{1})$, where,
\begin{equation}
J_{i}(t_{0},t_{1}) = \int_{t_{0}}^{t_{1}}R_{i}(t)dt = \frac{(t_{1}-t_{0})}{2}
\left[ 2R_{i}(t_{0})+\dot{R}_{i}(t_{0})(t_{1}-t_{0})\right].
\label{Eq:ContributionTarget}
\end{equation}
\end{theorem}

\paragraph{\textbf{Local Sensing Objective Function}}

The \emph{local sensing objective function} of a target $i\in\mathcal{T}$ over a period $[t_{0},t_{1})\subseteq [0,T]$ is defined as
\begin{equation}
\bar{J}_{i}(t_{0},t_{1})=\sum_{j\in\bar{\mathcal{N}}_{i}}J_{j}(t_{0},t_{1}),\label{Eq:LocalObjectiveFunction}%
\end{equation}
where each $J_{j}(t_{0},t_{1})$ term is evaluated using Theorem \ref{Th:ContributionTarget}.


\paragraph{\textbf{Decomposition of the Energy Objective $J_e$}} 

A similar decomposition result as Theorem \ref{Th:ContributionTarget} applies to the energy objective $J_e$ defined in \eqref{Eq:MainObjective}. However, this result is immediate from  \eqref{Eq:MainObjective} and is as follows. The contribution to the term $J_e$ in \eqref{Eq:MainObjective} by an agent $a\in\mathcal{A}$ from traversing a trajectory segment $(i,j) \in \mathcal{E}$ over the transit-time interval $[t_o,t_f]\triangleq\mathcal{P}_{ij}$\ \,is\  \,$J_a(t_o,t_f)$, where,
\begin{equation}
    J_a(t_o,t_f) = \int_{t_o}^{t_f}u_a^2(t)\,dt.
    \label{Eq:ContributionAgent}
\end{equation}
Note that the agent does not have any contribution to the $J_e$ term during dwell-time intervals as $u_a(t)=0$ during such periods.




\paragraph{\textbf{Agent Angular Velocity Profile} $w_a(t)$}
The control profile segment $\{w_a(t):t\in\mathcal{P}_{ij}\}$ that needs to be used by an agent $a\in\mathcal{A}$ over the transit-time interval $\mathcal{P}_{ij}$ on the trajectory segment $(i,j)\in\mathcal{E}$ can be obtained using only the following information: (i) the agent tangential acceleration profile $\{u_a(t):t\in \mathcal{P}_{ij}\}$ and (ii) the shape of the trajectory segment $(i,j)$ given in a parametric form $\{(x(p),y(p)):p\in[p_o,p_f]\}$. Note that the parameter values $p = p_o$ and $p = p_f$ correspond to the terminal target locations $Y_i \equiv (x(p_o),y(p_o))$ and $Y_j \equiv (x(p_f),y(p_f))$, respectively. For notational convenience, let us denote $x_p^\prime = \frac{dx(p)}{dp}$,\ $y_p^\prime = \frac{dy(p)}{dp}$,\ $x_p^{\prime\prime} = \frac{d^2x(p)}{dp^2}$ and $y_p^{\prime\prime} = \frac{d^2y(p)}{dp^2}$. 

First, we require a minor technical assumption regarding the said trajectory segment shape parameterization. 
\begin{assumption}\label{As:AngularVelocity}
There exists an injective (i.e., one-to-one) function $f:[p_o,p_f]\rightarrow[0,y_{ij}]$ such that   
\begin{equation}\label{Eq:ParametrizationFunction1}
    f(p) \triangleq \int_{p_o}^{p}\sqrt{(x_p^\prime)^2 + (y_p^\prime)^2}\,dp, 
\end{equation}
with $f(p_f) = y_{ij}$ and a corresponding inverse function $f^{-1}$.
\end{assumption}

This assumption simply means that we should be able to express the distance, say $l$, along the trajectory segment starting from $(x(p_o),y(p_o))$ to $(x(p),y(p))$ where $p\in[p_o,p_f]$, explicitly in terms of the parameter $p$ (i.e., $l=f(p)$) and vice versa (i.e., $p=f^{-1}(l)$). Clearly this assumption holds if the distance $l$ is used directly as the parameter $p$ (i.e., $p=l$) that characterizes the trajectory segment shape. 

Second, let us define a function $F:[p_o,p_f]\rightarrow\R$ such that
\begin{equation}\label{Eq:ParametrizationFunction2}
    F(p) \triangleq \frac{x'_p y''_p-y'_p x''_p}{\left((x'_p)^2 + (y'_p)^2\right)^{\frac{3}{2}}}.
\end{equation}

Finally, as shown in Fig. \ref{Fig:AgentTrajectory2}, let us denote by $l_a(t),\,t\in\mathcal{P}_{ij}$ the total distance the agent has traveled on the trajectory segment $(i,j)$ by time $t$. According to \eqref{Eq:AgentDynamics}, $v_a(t),\,t\in\mathcal{P}_{ij}$ represents the agent tangential velocity on the trajectory segment at time $t$. Considering the agent dynamics along the tangential direction to the trajectory segment, note that we can write 
\begin{equation}\label{Eq:AgentDynamicsInTangentialDirection0}
    l_a(t) = \int_{t_o}^{t} v_a(t)dt\ \mbox{ and }\ 
    v_a(t) = \int_{t_o}^{t} u_a(t)dt,
\end{equation}
for all $t\in[t_o,t_f]\triangleq \mathcal{P}_{ij}$ (note also that the terminal conditions $l_a(t_f) = y_{ij}$ and $v_a(t_f) = 0$ should be satisfied by \eqref{Eq:AgentDynamicsInTangentialDirection0}). 

\begin{figure}[!h]
\centering
\includegraphics[width=2.7in]{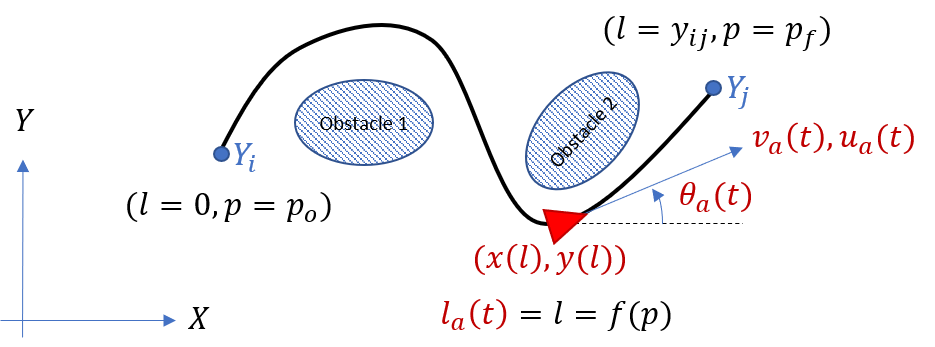} \caption{Angular velocity control of an agent $a\in\mathcal{A}$ while traversing a trajectory segment $(i,j)\in\mathcal{E}$.}%
\label{Fig:AgentTrajectory2}%
\end{figure}

\begin{theorem}\label{Th:AngularVelocity}
The required agent angular velocity profile $\{w_a(t): t\in\mathcal{P}_{ij}\}$ on trajectory segment $(i,j)\in\mathcal{E}$ is 
\begin{equation}\label{Eq:AngularVelocity}
    w_a(t) = F(f^{-1}(l_a(t)))\,v_a(t),
\end{equation}
where $f(\cdot)$ and $F(\cdot)$ are as in \eqref{Eq:ParametrizationFunction1} and \eqref{Eq:ParametrizationFunction2}, respectively.
\end{theorem}
\begin{proof}
Provided in Appendix \ref{App:AngularVelocityProof}.
\end{proof}

For an example, if the trajectory segment $(i,j)$ (between target locations $Y_i$ and $Y_j$) takes a circular shape centered at $C_{ij}\in\R^2$ with a radius $r_{ij}$, it can be represented by the parametric form: $\{(x(p),y(p)):p\in[p_o,p_f]\}$ where $(x(p),y(p)) \equiv C_{ij}+r_{ij} [\cos(p),\, \sin(p)]^T$ and $p_o = \arctan(Y_i-C_{ij})$, $p_f = \arctan(Y_j-C_{ij})$. Using \eqref{Eq:ParametrizationFunction2} and \eqref{Eq:AngularVelocity}, it can be shown that $F(p) = 1/r_{ij}$ and $w_a(t) = v_a(t)/r_{ij}$, respectively.

Similarly, if the trajectory segment $(i,j)$ takes a linear shape, it can be shown that 
$f(p)=p$ and $F(p) = 0$ from \eqref{Eq:ParametrizationFunction1} and  \eqref{Eq:ParametrizationFunction2}, respectively. Therefore, \eqref{Eq:AngularVelocity} reveals that $w_a(t) = 0$.


\begin{remark}
In robotics applications where line-following techniques can be used \cite{Pakdaman2009}, an agent can use its line-following capabilities to control its angular velocity $w_a(t)$ (instead of using \eqref{Eq:AngularVelocity}) - irrespective of its tangential acceleration $u_a(t)$.
\end{remark}

In conclusion, Theorem \ref{Th:AngularVelocity} allows us to dispense of $w_a(t)$ as a control input because it is always determined through $u_a(t)$ (which gives $v_a(t)$ via \eqref{Eq:AgentDynamicsInTangentialDirection0}) and the prespecified shape of the trajectory segment.

\paragraph{\textbf{The Equivalent Dynamic Agent Model}}

Since we now have discussed how an agent $a\in\mathcal{A}$ can control its angular velocity $w_a(t)$ (i.e., via \eqref{Eq:AngularVelocity}), we can omit angular dynamics from \eqref{Eq:AgentDynamics} to construct an equivalent dynamic agent model, that focuses only on the tangential dynamics on a trajectory segment $(i,j)\in\mathcal{E}$. In particular, as a direct consequence of \eqref{Eq:AgentDynamicsInTangentialDirection0}, upon taking the state vector as $[l_a(t),\ v_a(t)]^T$ for some $t\in\mathcal{P}_{ij}$, we can express the corresponding state dynamics as a second-order single-input linear system:
\begin{equation}\label{Eq:AgentDynamics2}
    \begin{bmatrix}
    \dot{l}_a(t)\\
    \dot{v}_a(t)
    \end{bmatrix}
    =
    \begin{bmatrix}
    0 & 1\\
    0 & 0
    \end{bmatrix}
    \begin{bmatrix}
    l_a(t)\\
    v_a(t)
    \end{bmatrix}
    +
    \begin{bmatrix}
    0\\
    1
    \end{bmatrix}
    u_a(t).
\end{equation}
In the sequel, we use \eqref{Eq:AgentDynamics2} and \eqref{Eq:AngularVelocity} to determine the optimal agent control profile segments $\{u_a(t):t\in\mathcal{P}_{ij}\}$ and $\{w_a(t):t\in\mathcal{P}_{ij}\}$, respectively. This particular decomposition of unicycle agent dynamics is fundamentally similar to that proposed in \cite{Kim2020}.

\subsection{ED-RHC Problem (RHCP) Formulation}
\label{SubSec:RHCIntro}

Consider an agent $a\in\mathcal{A}$ residing on a target $i\in\mathcal{T}$ at some time $t_s\in[0,T]$. Recall that control $U_{ia}(t_s)$ in \eqref{RHC problem} includes \emph{dwell-time} decisions $\delta_{i}$ and $\delta_j$ at the current target $i$ and the next-visit target $j\in\mathcal{N}_i$,  respectively. As shown in Fig. \ref{Fig:Timeline}, a dwell-time decision $\delta_{i}$ (or $\delta_{j}$) can be divided into two interdependent decisions: 
(\romannum{1}) the \emph{active time} $\tau_{i}$ (or $\tau_{j}$) and
(\romannum{2}) the \emph{inactive }(or\emph{ idle}) \emph{time} $\bar{\tau}_i$ (or $\bar{\tau}_j$). 
Therefore, the agent has to optimally choose \emph{decision variables} which form the control vector $U_{ia}(t_s)=[\tau_{i},\bar{\tau}_i,j,\{u_a(t)\}, \tau_{j},\bar{\tau}_j]$. Note that here we have: 
(\romannum{1}) omitted representing each of these decision variable's dependence on $t_s$, 
(\romannum{2}) used the notation $\{u_a(t)\}$ to represent $\{u_a(t):t\in\mathcal{P}_{ij}(t_s)\}$ and (\romannum{3}) omitted $\{w_a(t):t\in\mathcal{P}_{ij}(t_s)\}$ as it can be found directly from $\{u_a(t)\}$ and \eqref{Eq:AngularVelocity}.

\paragraph{\textbf{The Receding Horizon Control Problem (RHCP)}}
Let us denote the real-valued component of the control vector $U_{ia}(t_s)$ in \eqref{RHC problem} as $U_{iaj}(t_s)=[\tau_{i},\bar{\tau}_i,\{u_a(t)\},\tau_{j},\bar{\tau}_j]$. The discrete component of $U_{ia}(t_s)$ is simply the next-visit target $j\in\mathcal{N}_{i}$. 
In this setting (see also Fig. \ref{Fig:Timeline}), we define the planning horizon length $\mathsf{w}(U_{ia}(t_s))$ in \eqref{RHC problem} as 
\begin{equation}\label{Eq:VariableHorizon}
    \mathsf{w}(U_{iaj}(t_s)) \triangleq \tau_i + \bar{\tau}_i + \rho_{ij} + \tau_j + \bar{\tau}_j. 
\end{equation}
The current local state $X_{ia}(t_s)$ in \eqref{RHC problem} is considered as $X_{ia}(t_s)=[s_a,v_a,\theta_a,\{R_{j}:j\in\bar{\mathcal{N}}_{i}\}]$ (again, omitting the dependence on $t_s$). Then, the optimal controls are obtained by solving \eqref{RHC problem}, which can be re-stated as the following set of optimization problems, henceforth called the RHC Problem (RHCP):
\begin{align}
&\begin{alignedat}{4}\label{Eq:RHCGenSolStep1}
& U_{iaj}^{\ast} = & \underset{\makebox[2cm]{\footnotesize $U_{iaj}\in\mathbb{U}$}}{\arg\min} & \quad && J_{H}(X_{ia}(t_s),U_{iaj};\ \mathsf{w}(U_{iaj}));\ \forall j\in\mathcal{N}_{i},\ & \\
& & \makebox[2cm]{subject to} &  && \mathsf{w}(U_{iaj}) \leq H & \\
\end{alignedat}\\
&\begin{alignedat}{4}\label{Eq:RHCGenSolStep2}
&\ \ j^{\ast} = & \underset{\makebox[2cm]{\footnotesize $j\in\mathcal{N}_{i}$}}{\arg\min} & \quad && J_{H}(X_{ia}(t_s),U_{iaj}^{\ast}; \ \mathsf{w}(U_{iaj}^*)). &
\end{alignedat}
\end{align}
Note that \eqref{Eq:RHCGenSolStep1} requires solving $\vert \mathcal{N}_i \vert$ optimization problems, one for each neighboring target $j\in\mathcal{N}_i$ ($\vert \cdot\vert$ is the cardinality operator). The next step \eqref{Eq:RHCGenSolStep2} is a simple comparison to determine the optimal next-visit target $j^*$. Therefore, the final optimal controls of the RHCP are
$U_{ia}^*(t_s)=[U_{iaj^{\ast}}^{\ast},j^*]$.

The objective function $J_{H}(\cdot)$ in \eqref{Eq:RHCGenSolStep1} is chosen to reflect the contribution to the main objective $J_{T}$ in \eqref{Eq:MainObjective} by the targets in the neighborhood $\bar{\mathcal{N}}_{i}$ and by the agent $a$, over the planning horizon $[t_s,\,t_s+\mathsf{w}]$ as
\begin{equation}
J_{H}(X_{ia}(t_s),U_{iaj};\ \mathsf{w}) 
\triangleq \alpha_H \underbrace{J_a(t_o,t_f)}_{\triangleq\, J_{eH}} + \underbrace{\frac{1}{\mathsf{w}}\bar{J}_{i}(t_s,t_s+\mathsf{w})}_{\triangleq\, J_{sH}}.
\label{J_H}
\end{equation}
where $w = \mathsf{w}(U_{iaj})$ and $\alpha_H \triangleq \alpha$ (the weight factor used in \eqref{Eq:MainObjective}). In \eqref{J_H}, the form of the $J_{sH}$ component has been selected so that it is analogous to the $J_s$ component in \eqref{Eq:MainObjective} (with $T$ replaced by $w$). As illustrated in Fig. \ref{Fig:Timeline}, note also that $t_e \triangleq t_s + \mathsf{w}$,\ \ $[t_o,t_f]\triangleq \mathcal{P}_{ij}\subseteq[t_s,t_e]$ and $\rho_{ij} \triangleq t_f-t_o$.

\begin{figure}[!h]
\centering
\includegraphics[width=3.2in]{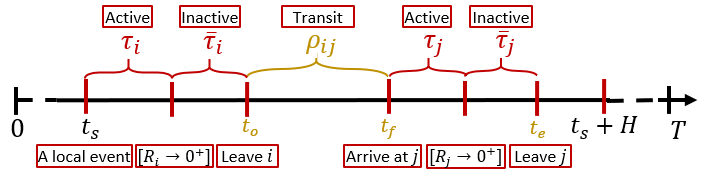} \caption{Event timeline and control decisions in ED-RHC.}%
\label{Fig:Timeline}%
\end{figure}

\paragraph{\textbf{Planning Horizon}}
In conventional RHC methods, the RHCP objective function is evaluated over a fixed planning horizon length $H$, where $H$ is selected exogenously. This makes the RHCP solution dependent on the choice of $H$. In contrast, through \eqref{J_H} and \eqref{Eq:VariableHorizon} above, we have made the RHCP solution (i.e., \eqref{Eq:RHCGenSolStep1} and \eqref{Eq:RHCGenSolStep2}) free of the parameter $H$, by using $H$ only as an upper-bound to the actual planning horizon length $\mathsf{w}(U_{iaj})$ in \eqref{Eq:VariableHorizon} and selecting $H$ to be sufficiently large (e.g., $H=T-t_s$).

In fact, since the planning horizon length $\mathsf{w}(U_{iaj})$ is a control variable, 
the above RHCP formulation simultaneously determines the \emph{optimal planning horizon} length $w^{\ast} = \mathsf{w}(U_{iaj^*}^*)$.  
Moreover, as shown in Fig. \ref{Fig:Timeline}, the time to depart from the current target $i$ (i.e., $t_o$), the time to arrive at the destination target $j$ (i.e., $t_f$) and the corresponding transit-time $\rho_{ij}=t_f-t_o$, are also control dependent. Hence, this RHCP formulation also determines the optimal values of each of these quantities: $t_o^*,\,t_f^*$ and $\rho_{ij^*}^*$, respectively.


\paragraph{\textbf{Overview of the RHCP Solution Process}}\label{Par:OverallSolution}
Looking back at \eqref{Eq:ContributionTarget} and \eqref{Eq:LocalObjectiveFunction}, notice that the sensing component $J_{sH}$ of the RHCP objective \eqref{J_H} does not explicitly depend on the agent control profile segment $\{u_a(t):t\in\mathcal{P}_{ij}\}$, but, it depends on the agent's transit-time $\rho_{ij}$ value and on the other control decisions in $U_{iaj}$: $\tau_i,\bar{\tau}_i,\tau_j,\bar{\tau}_j$. Therefore, let us denote $J_{sH}$ as a function parameterized by $\rho_{ij}$:  $J_{sH}(\tau_i,\bar{\tau}_i,\tau_j,\bar{\tau}_j;\,\rho_{ij})$.  

In contrast, based on \eqref{Eq:ContributionAgent}, notice that the energy component $J_{eH}$ of the RHCP objective $\eqref{J_H}$ only depends on agent control profile segments, specifically on $\{u_a(t):t\in\mathcal{P}_{ij}\}$. Therefore, let us denote $J_{eH}$ simply as $J_{eH}(\{u_a(t)\})$.   

As illustrated in Fig. \ref{Fig:OverviewOfTheRHCPSolution}, we exploit this property of the RHCP objective components ($J_{sH}$ and $J_{eH}$) to solve the RHCP \eqref{Eq:RHCGenSolStep1}. 
In particular, we start with analytically solving the optimization problem which we label as the RHCP($\rho_{ij}$):
\begin{equation}\label{Eq:RHCP_rho_ij}
\begin{alignedat}{4}
& J_{sH}^*(\rho_{ij}) \triangleq & \underset{\makebox[2.7cm]{\footnotesize $(\tau_i,\bar{\tau}_i,\tau_j,\bar{\tau}_j)\in\mathbb{U}_{s}(\rho_{ij})$}}{\min} & \quad && J_{sH}(\tau_i,\bar{\tau}_i,\tau_j,\bar{\tau}_j;\,\rho_{ij}). &
\end{alignedat}
\end{equation}
For this purpose, we exploit a few results established in \cite{Welikala2020J4} where the RHCP($\rho_{ij}$) \eqref{Eq:RHCP_rho_ij} has already been solved while treating $\rho_{ij}$ as a known constant.


Next, we use the $J_{sH}^*(\rho_{ij})$ function obtained from \eqref{Eq:RHCP_rho_ij} and the relationship $\rho_{ij}=t_f-t_o$ to reformulate the problem of optimizing the RHCP objective \eqref{J_H} as an optimal control problem (OCP):
\begin{equation}\label{Eq:OCP}
\begin{alignedat}{4}
& [t_o^*,t_f^*,\{u_a^*(t)\}] = & \underset{\makebox[1.7cm]{\footnotesize $t_o,\,t_f,\,\{u_a(t)\}$}}{\arg\min} & \ && \alpha_H J_{e}(\{u_a(t)\}) + J_{sH}^*(t_f-t_o). &
\end{alignedat}
\end{equation}

Finally, as shown in Fig. \ref{Fig:OverviewOfTheRHCPSolution}, it is straightforward how the RHCP \eqref{Eq:RHCGenSolStep1} solution $U_{iaj}^*$ can be constructed from the obtained solutions of the OCP \eqref{Eq:OCP} and the RHCP($\rho_{ij}$) \eqref{Eq:RHCP_rho_ij}.

\begin{figure}[!h]
    \centering
    \includegraphics[width=3.2in]{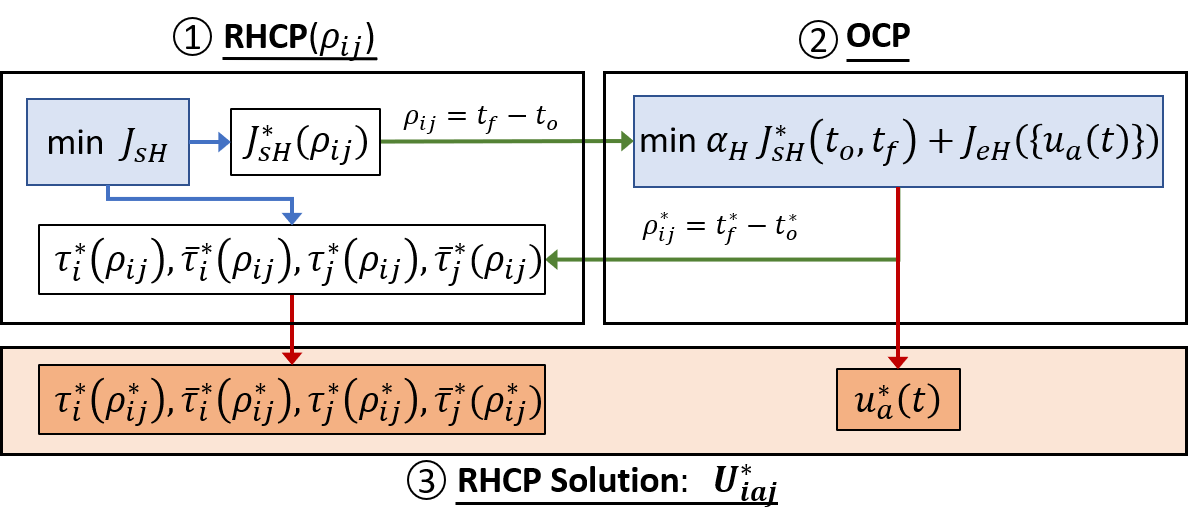}
    \caption{Overview of the RHCP Solution Process when solving \eqref{Eq:RHCGenSolStep1} for some next-visit target $j\in\mathcal{N}_i$. 1. Solving the receding horizon control component  of \eqref{Eq:RHCGenSolStep1} (i.e., \eqref{Eq:RHCP_rho_ij});  
    2. Solving the optimal control component of \eqref{Eq:RHCGenSolStep1} (i.e., \eqref{Eq:OCP});
    3. Constructing the final solution of \eqref{Eq:RHCGenSolStep1}.}
    \label{Fig:OverviewOfTheRHCPSolution}
\end{figure}

\paragraph{\textbf{Event-Driven Action Horizon}}
Each RHCP solution (i.e., $U_{ia}^*(t_s) = [U_{iaj^*}^*, j^*]$ from \eqref{Eq:RHCGenSolStep1}-\eqref{Eq:RHCGenSolStep2}) obtained over a planning horizon $\mathsf{w}(U_{iaj^*}^*)\leq H$ is generally executed over a shorter \emph{action horizon} $h\leq \mathsf{w}(U_{iaj^*}^*)$. In particular, the action horizon $h$ is determined by the first event that takes place after $t_s$, where the RHCP was last solved. Such a subsequent event may be \emph{controllable} if it results from executing the last solved RHCP solution or \emph{uncontrollable} if it results from a random or an external event (if such events are allowed). 

When executing the RHCP solution obtained by an agent at target $i$ at time $t_s$, there are three mutually exclusive controllable events that may occur subsequently. They are:

\paragraph*{\textbf{1. Event }$[h\rightarrow \tau_{i}^{\ast}]$} This event is feasible only if $\tau_{i}^{\ast}(t_s)>0$ and it occurs at a time $t = t_s+\tau_{i}^{\ast}(t_s)$. If $R_i(t)>0$, it coincides with a departure event from target $i$. Otherwise, i.e., if $R_i(t)=0$, it coincides with a $[R_i\rightarrow 0^+]$  event. 

\paragraph*{\textbf{2. Event }$[h\rightarrow \bar{\tau}_i^{\ast}]$} This event is feasible if $\tau_{i}^{\ast}(t_s)=0$ (when $R_i(t_s)=0$) and $\bar{\tau}_i^{\ast}(t_s)\geq0$. It occurs at $t=t_s+\bar{\tau}_i^{\ast}(t_s)$ and coincides with a departure event from target $i$.

\paragraph*{\textbf{3. Event }$[h\rightarrow\rho_{ij^{\ast}}]$} This event is feasible only if a departure event (from target $i$) occurred at $t_s$. Clearly this event coincides with an arrival event at target $j^*(t_s)$.  

In an agent trajectory, at a given time instant, only one of these three controllable events is feasible. However, there are two uncontrollable events that may occur at an agent residing in a target $i$ due to two specific controllable events at a neighboring target $j\in\mathcal{N}_i$. These two types of events are aimed to enforce the ``no simultaneous target sharing'' condition (i.e., the control constraint \eqref{Eq:NoTargetSharing}) and thus, only applies to multi-agent problems. 
To enforce this condition, an agent at target $i$ modifies its neighborhood $\mathcal{N}_i$ to $\mathcal{N}_i\backslash\{j\}$ when:   
(i) another agent already resides at target $j$ or (ii) another agent is en-route to visit target $j$. Therefore, we define the following two \emph{neighbor induced events} at target $i$ due to a neighbor $j\in\mathcal{N}_i$:

\paragraph*{\textbf{4. Covering Event} $C_{j},\,j\in\mathcal{N}_i$} This event causes $\mathcal{N}_i$ to be modified to $\mathcal{N}_i\backslash\{j\}$. 
\paragraph*{\textbf{5. Uncovering Event} $\bar{C}_{j},\,j\in\mathcal{N}_i$} This event causes $\mathcal{N}_i$ to be modified to $\mathcal{N}_i\cup\{j\}$.

If one of these two events occurs while the agent is awaiting an event $[h\rightarrow\tau_i^*]$ or $[h\rightarrow\bar{\tau}_i^*]$, the RHCP is resolved to account for the updated neighborhood $\mathcal{N}_i$. 

\paragraph{\textbf{Three Forms of RHCPs}}
The exact form of the RHCP (\eqref{Eq:RHCGenSolStep1} and \eqref{Eq:RHCGenSolStep2}) that needs to be solved at a certain time depends on the event that triggered the end of the previous action horizon. In particular, corresponding to the three controllable event types, there are three forms of RHCPs:

\textbf{RHCP1:} 
At a target $i$ and time $t_s$, this particular problem form is solved upon: (i) an arrival event $[h\rightarrow\rho_{ki}]$ where $k\in\mathcal{N}_i$ or (ii) a $C_j$ (or a $\bar{C}_j$) event occurred when $R_i(t_s)>0$ where $j\in\mathcal{N}_i$. Since $R_i(t_s)>0$ whenever this problem form is solved, it is equivalent to the generic form of the RHCP that needs to be solved for the complete set of decision variables: $U_{ia}(t_s) = [\tau_i,\bar{\tau}_j,j,\{u_a(t)\},\tau_j,\bar{\tau}_j]$ with $\tau_i \geq 0$.


\textbf{RHCP2:} 
At a target $i$ and time $t_s$, this particular problem form is solved when $R_i(t_s)=0$ upon: (i) an event $[h\rightarrow\tau_i^*]$ or (ii) a $C_j$ (or a $\bar{C}_j$) event where $j\in\mathcal{N}_i$. Since $R_i(t_s)=0$ whenever this problem form is solved, it is the same as \textbf{RHCP1} but with $\tau_{i} = 0$, hence simpler.


\textbf{RHCP3:} 
At a target $i$ and time $t_s$, this particular problem form is solved upon: (i) an event $[h\rightarrow\tau_i^*]$ with $R_i(t_s)>0$ or (ii) an event $[h\rightarrow\bar{\tau}_i^*]$. Simply, this problem form is solved whenever the agent is ready to depart from the target. Therefore, it is the same as \textbf{RHCP1} but with $\tau_{i} = 0$ and $\bar{\tau}_i = 0$.



\section{Solving Event-Driven RHCPs}
\label{Sec:SolutionToTheRHCPs}

In this section, we present the solutions to the three RHCP forms identified above. We begin with \textbf{RHCP3}.

\subsection{Solution of \textbf{RHCP3}}
\textbf{RHCP3} is the simplest RHCP given that $\tau_i = \bar{\tau}_i = 0$ in $U_{ia}$ by default. Therefore, $U_{iaj}$ (i.e., the real-valued component of $U_{ia}$\,, used in \eqref{Eq:RHCGenSolStep1}) is limited to $U_{iaj}=[\{u_a(t)\},\tau_{j},\bar{\tau}_j]$ and the planning horizon $\mathsf{w}(U_{iaj})$ defined in \eqref{Eq:VariableHorizon} becomes $\mathsf{w}(U_{iaj})=\rho_{ij}+\tau_{j}+\bar{\tau}_j$. 

Under these conditions, we next solve \eqref{Eq:RHCGenSolStep1} (via solving RHCP($\rho_{ij}$) \eqref{Eq:RHCP_rho_ij} and OCP \eqref{Eq:OCP}, as shown in Fig. \ref{Fig:OverviewOfTheRHCPSolution}) and \eqref{Eq:RHCGenSolStep2} to obtain the \textbf{RHCP3} solution.


\paragraph{\textbf{Solution of RHCP($\rho_{ij}$) \eqref{Eq:RHCP_rho_ij}}}
As mentioned before, RHCP($\rho_{ij}$) has already been solved in \cite{Welikala2020J4} - while treating $\rho_{ij}$ as a known fixed value. In particular, the RHCP($\rho_{ij}$) solution corresponding to the \textbf{RHCP3} takes the form \cite[Th.~2]{Welikala2020J4}: 
\begin{equation}\label{Eq:OP3FormalArgMinSolution}
\begin{aligned}
(\tau_i^*, \bar{\tau}_i^*) =&\ (0,0),\\
(\tau_{j}^{*},\bar{\tau}_j^{*}) =&\  
\begin{cases}
(0,0) & \mbox{ if } \bar{A} \geq B_{j} \mbox{ or } D_1 > D_2\\
(D_2,0) & \mbox{ else if } D_2 < D_3\\
(D_3,0) & \mbox{ else if } B_{j} > \bar{A} \geq B_{j}\left[
1-\frac{\rho_{ij}^{2}}{(\rho_{ij}+D_3)^{2}}\right] \\
(D_3,\bar{D}_1) & \mbox{ else if } \bar{D}_1 \leq \bar{D}_2\\
(D_3,\bar{D}_2) & \mbox{ otherwise, }
\end{cases}
\\
J_{sH}^*(\rho_{ij}) =&\ J_{sH}(\tau_{j}^*,\bar{\tau}_j^*;\,\rho_{ij}),
\end{aligned}
\end{equation}
where 
\begin{equation}\label{Eq:OP3FormalArgMinSolutionCoeffs}
\begin{gathered}
\bar{A} = \sum_{m\in\bar{\mathcal{N}}_{i}}A_{m},\ 
D_1 = \frac{\bar{A}\rho_{ij}}{B_{i}-\bar{A}},\ 
D_2 = \min\{D_3,\ H-\rho_{ij}\},\\  
D_3 = \frac{R_{j}(t_o)}{B_{j}-A_{j}}+\frac{A_{j}}{B_{j}-A_{j}}\rho_{ij},\\
\bar{D}_1 = \sqrt{\frac{(B_{j}-A_{j})(\rho_{ij}+D_3)^{2}-B_{j}\rho_{ij}^{2}}{\bar{A}_{j}}}-(\rho_{ij}+D_3), \\
\bar{A}_{j} = \bar{A}-A_{j},\ 
\bar{D}_2\triangleq H-(\rho_{ij}+D_3),\\
J_{sH}(\tau_j,\bar{\tau}_j;\,\rho_{ij})=\frac{C_{1}\tau_{j}^{2}+C_{2}\bar{\tau}_j^{2}+C_{3}\tau_{j}\bar{\tau}_j%
+C_{4}\tau_{j}+C_{5}\bar{\tau}_j+C_{6}}{\rho_{ij}+\tau_{j}+\bar{\tau}_j},\\
C_{1}=\frac{1}{2}[\bar{A}-B_{j}],\  C_{2}=\frac{\bar{A}_{j}}{2},\ 
C_{3}=\bar{A}_{j},\  C_{4}=[\bar{R}(t_o)+\bar{A}\rho_{ij}],\\
\bar{R} = \sum_{m\in\bar{\mathcal{N}}_{i}}R_{m},\ 
C_{5}=[\bar{R}_{j}(t_o)+\bar{A}_{j}\rho_{ij}],\ 
\bar{R}_{j} \triangleq \bar{R}-R_{j},\\
C_{6}=\frac{\rho_{ij}}{2}[2\bar{R}(t_o)+\bar{A}\rho_{ij}].
\end{gathered}
\end{equation}


Note that in \eqref{Eq:OP3FormalArgMinSolution}, not only $J_{sH}^*$, but also $\tau_j^*$ and $\bar{\tau}_j^*$ are functions of the transit-time $\rho_{ij}$. To provide intuition about the $J_{sH}^*(\rho_{ij})$ function form, let us consider the first case in \eqref{Eq:OP3FormalArgMinSolution} where $(\tau_j^*,\bar{\tau}_j^*)=(0,0)$ that results in   
\begin{equation}
    J_{sH}^*(\rho_{ij}) = J_{sH}(0,0;\,\rho_{ij}) =  \bar{R}(t_o)+\frac{1}{2}\bar{A}\rho_{ij},
\end{equation}
under the condition $\bar{A} \geq B_{j}$ or $D_1 > D_2$. Using \eqref{Eq:OP3FormalArgMinSolutionCoeffs}, it can be shown that 
$$ D_1 > D_2 \ \iff \ \rho_{ij} > \min 
\{
\frac{R_j(t_o)(B_i-\bar{A})}{\bar{A}B_j-A_jB_i},\ 
H(1-\frac{\bar{A}}{B_i})
\}.
$$
From this example, it is clear that the function $J_{sH}^*(\rho_{ij})$ is dependent on the neighborhood parameters (e.g., $\bar{A}, B_j, B_i$) as well as the current neighborhood state (e.g., $\bar{R}(t_o)$, $R_j(t_o)$).

\paragraph{\textbf{Objective Function of OCP \eqref{Eq:OCP}}}
Note that we now have solved the RHCP($\rho_{ij}$) and have obtained the functions (of $\rho_{ij}$): $\tau_i^*, \bar{\tau}_i^*, \tau_j^*, \bar{\tau}_j^*$, and, most importantly, $J_{sH}^*$. Based on the RHCP solution process outlined in Fig. \ref{Fig:OverviewOfTheRHCPSolution}, our next step is to formulate and solve the corresponding OCP \eqref{Eq:OCP}.  

As shown in \eqref{Eq:OCP}, the sensing objective component of OCP is $J_{sH}^*(t_f-t_o)$. Note that we now can explicitly express this term using the obtained $J_{sH}^*(\rho_{ij})$ function in \eqref{Eq:OP3FormalArgMinSolution} and the relationship $\rho_{ij}=t_f-t_o$. For notational convenience, taking into account that in \textbf{RHCP3}, $t_o$ is the the current event time when the RHCP is solved (i.e., $t_o=t_s$ where $t_s$ is fixed and known), let us denote this sensing objective component of OCP as
\begin{equation}\label{Eq:ORHCP3SensingComponent}
    \phi(t_f) \triangleq J_{sH}^*(t_f-t_o).
\end{equation}

On the other hand, using \eqref{J_H} and \eqref{Eq:ContributionAgent}, the energy objective component of OCP \eqref{Eq:OCP} can be expressed as
\begin{equation}\label{Eq:ORHCP3EnergyComponent}
    J_{eH}(\{u_a(t)\}) = \int_{t_o}^{t_f}u_a^2(t)dt.
\end{equation}

\paragraph{\textbf{Solution of OCP \eqref{Eq:OCP}}}
In the following analysis, for notational convenience, we use $\dot{x} = Ax(t) + Bu(t)$ with 
\begin{equation}\label{Eq:AgentDynamicsCoeffs}
A = 
\begin{bmatrix}
0 & 1\\ 0 & 0
\end{bmatrix}
,\ \ 
B = 
\begin{bmatrix}
0 \\ 1
\end{bmatrix}
,\ \ 
x(t) = 
\begin{bmatrix}
l_a(t) \\
v_a(t)
\end{bmatrix}
,\ \ 
u(t) = u_a(t),    
\end{equation}
to represent the agent dynamics stated in \eqref{Eq:AgentDynamics2}. Under this notation, using \eqref{Eq:ORHCP3SensingComponent} and \eqref{Eq:ORHCP3EnergyComponent}, the OCP \eqref{Eq:OCP} can be stated as 
\begin{equation}\label{Eq:ORHCP3AgentTrajectoryOptimization}
\begin{alignedat}{3}
& \underset{\makebox[2cm]{\footnotesize $t_f,\{u(t)\}$}}{\min}   & \quad &  \alpha_H\int_{t_o}^{t_f}u^2(t)dt + \phi(t_f)\\
& \makebox[2cm]{subject to} &       & \dot{x} = Ax(t) + Bu(t),\\
&  &  & x(t_o) = [0, 0]^T,\ \ x(t_f) = [y_{ij}, 0]^T.
\end{alignedat}
\end{equation}
The last two constraints in \eqref{Eq:ORHCP3AgentTrajectoryOptimization} are simply terminal constraints for the agent motion on the trajectory segment $(i,j)$. Note that \eqref{Eq:ORHCP3AgentTrajectoryOptimization} is a standard free final time, fixed initial and final state optimal control problem. Hence, there is an established solution procedure \cite{bryson1975} as outlined next.

First, the Hamiltonian corresponding to  \eqref{Eq:ORHCP3AgentTrajectoryOptimization} is written as 
\begin{equation}\label{Eq:RHCP3OCPHamiltonian}
    H(x(t),u(t),t) \triangleq \alpha_H u^2(t) + \lambda^T(t) (Ax(t) + Bu(t)),
\end{equation}
where $\lambda(t)$ represents the co-state variables. Next, the adjoined function that combines the terminal constraint on $x(t_f)$ and the terminal cost $\phi(t_f)$ is written as
$$\Phi(x(t_f),t_f) \triangleq \phi(t_f) + \nu^T(x(t_f)-[y_{ij}, 0]^T),$$ 
where $\nu$ is a set of multipliers. 

Finally, the OCP in \eqref{Eq:ORHCP3AgentTrajectoryOptimization} can be solved to obtain the corresponding optimal   $\{x(t),u(t),\lambda(t): t\in[t_o,t_f]\},\,t_f$ and $\nu$ values by solving the following system of equations \cite{bryson1975}:
\begin{gather}
\label{Eq:OptimalControlEquations3}
\frac{\partial H}{\partial u} = 2\alpha_Hu(t) + \lambda^T(t)B = 0,\\
\label{Eq:OptimalControlEquations4}
\dot{\lambda} = -\left(\frac{\partial H}{\partial x}\right)^T = -A^T\lambda(t),\ \ \lambda(t_f) = \left(\frac{\partial \Phi}{\partial x(t_f)}\right)^T = \nu,\\
\label{Eq:OptimalControlEquations5}
\frac{d\Phi}{dt_f} + \alpha_H u^2(t_f) = 
\frac{d\phi}{dt_f} + \nu^T(Ax(t_f) + Bu(t_f)) + \alpha_H u^2(t_f) 
= 0 
\end{gather}
in addition to the agent dynamics and terminal constraints given in \eqref{Eq:ORHCP3AgentTrajectoryOptimization}.
Note that \eqref{Eq:OptimalControlEquations3} is the optimality condition (from Pontryagin's minimum principle), \eqref{Eq:OptimalControlEquations4} are the co-state equations and \eqref{Eq:OptimalControlEquations5} is the transversality condition.



\begin{lemma}\label{Lm:RHCP3OCPSolution}
The optimal terminal time $t_f^*$ of the OCP \eqref{Eq:ORHCP3AgentTrajectoryOptimization} satisfies the equation:
\begin{equation}\label{Eq:RHCP3OCPSolutiont_f}
(t_f-t_o)^4 \, \frac{d\phi(t_f)}{d t_f} = 36 \alpha_H y_{ij}^2,
\end{equation}
where $\phi(t_f)$ is known from \eqref{Eq:ORHCP3SensingComponent} and the corresponding optimal control law $u^*(t)$ is given by
\begin{equation}\label{Eq:RHCP3OCPSolutionu_t}
    u^*(t) = \frac{12 y_{ij}}{(t_f^*-t_o)^3}\left[ \frac{t_f^*+t_o}{2} - t\right],\ \forall t \equiv [t_o,t_f].
\end{equation}
\end{lemma}

\begin{proof}
First, we take $\nu = [\nu_1, \nu_2]^T$, $\lambda(t) = [\lambda_1(t), \lambda_2(t)]^T$ and solve \eqref{Eq:OptimalControlEquations4} for $\lambda_1(t)$ and $\lambda_2(t)$. This gives: 
\begin{equation*}
\lambda_1(t) = \nu_1\ \mbox{ and }\ 
\lambda_2(t) = \nu_2 + \nu_1(t_f-t),\ \forall t\in \mathcal{P}_{ij},
\end{equation*}
(recall $\mathcal{P}_{ij} =[t_o,t_f]$). We then solve \eqref{Eq:OptimalControlEquations3} for $u(t)$ to obtain: \begin{equation}\label{Eq:RHCP3OCPSolutionu_tProofStep1}
    u(t) = -\frac{\lambda_2(t)}{2\alpha_H} = -\frac{1}{2\alpha_{H}}\left(\nu_2 + \nu_1(t_f-t)\right),\ \forall t\in\mathcal{P}_{ij}.
\end{equation}

Next, we take $x(t) = [x_1(t), x_2(t)]^T$ and solve the agent dynamics equation in \eqref{Eq:ORHCP3AgentTrajectoryOptimization} (also using \eqref{Eq:RHCP3OCPSolutionu_tProofStep1}) for $x_2(t)$. This results in:
\begin{equation*}
    x_2(t) = \frac{(t_o-t)}{2\alpha_H}
    \left(
    \nu_2 + \nu_1t_f - \frac{\nu_1}{2}(t_o + t)
    \right),\ \forall t\in\mathcal{P}_{ij}.
\end{equation*}
Now, using the terminal constraint  $x_2(t_f)=0$ on the above, we get $\nu_2 = -\nu_1(t_f-t_o)/2$. Back substituting this in above $x_2(t)$ we get a further simplified expression for it as 
\begin{equation*}
    x_2(t) =  \frac{\nu_1}{4\alpha_H}
    \left(
    t^2 - (t_o+t_f)t + t_ot_f\right),\ \forall t\in\mathcal{P}_{ij}.
\end{equation*}
Applying this result in the relationship $x_1(t) = \int_{t_o}^{t}x_2(t)dt$ (i.e., agent dynamics) we get:
\begin{equation*}
x_1(t) = \frac{v_1(t-t_o)}{24\alpha_{H}}
\left(
2t^2 - (3t_f+t_o)t + t_o(3t_f-t_o)
\right),\ \forall t\in\mathcal{P}_{ij}.
\end{equation*}
Similar to before, using the terminal constraint $x_1(t_f)=y_{ij}$ on the above (and via back substituting), we get
\begin{equation}\label{Eq:RHCP3OCPSolutionu_tProofStep2}
    \nu_1 = \frac{24\alpha_Hy_{ij}}{(t_f-t_o)^3} 
    \ \mbox{(and }
    \nu_2 = -\frac{12\alpha_Hy_{ij}}{(t_f-t_o)^2},\ 
    u(t_f) = \frac{6y_{ij}}{(t_f-t_o)^2}
    \mbox{)}.
\end{equation}

Now we are ready to use \eqref{Eq:OptimalControlEquations5} to solve for the optimal $t_f$ value (i.e., $t_f^*$). Note that \eqref{Eq:OptimalControlEquations5} directly simplifies to the form:
\begin{equation*}
    \frac{d\phi}{dt_f} + \nu_2u(t_f) + \alpha_{H}u^2(t_f)=0,
\end{equation*}
which we can further reduce to the form (using  \eqref{Eq:RHCP3OCPSolutionu_tProofStep2}):
\begin{equation*}
    \frac{d\phi}{dt_f} -\frac{36\alpha_H y_{ij}}{(t_f-t_o)^4} = 0,
\end{equation*}
and obtain \eqref{Eq:RHCP3OCPSolutiont_f}. Finally, the optimal control law $u^*(t)$ in  \eqref{Eq:RHCP3OCPSolutionu_t} can be obtained by substituting \eqref{Eq:RHCP3OCPSolutionu_tProofStep2} in \eqref{Eq:RHCP3OCPSolutionu_tProofStep1}.
\end{proof}

Using the optimal terminal time $t_f^*$ and control $u^*(t)$ (i.e., $u_a^*(t)$) proven in Lemma \ref{Lm:RHCP3OCPSolution}, the optimal energy objective component of this OCP (i.e., \eqref{Eq:ORHCP3EnergyComponent}) can be obtained as 
\begin{equation}\label{Eq:RHCP3EnergySolution}
    J_{eH}(\{u_a^*(t)\}) = \frac{12 y_{ij}^2}{(t_f^*-t_o)^3}.
\end{equation}
The corresponding optimal sensing objective component (i.e., \eqref{Eq:ORHCP3SensingComponent}) is directly given by $\phi(t_f^*) = J_{sH}^*(t_f^*-t_o)$. Finally, the optimal transit-time value is $\rho_{ij}^* = t_f^* - t_o$.

\paragraph{\textbf{Solution of RHCP \eqref{Eq:RHCGenSolStep1} for} $U_{iaj}^*$}
As outlined in Fig. \ref{Fig:OverviewOfTheRHCPSolution}, we now can conclude solving RHCP \eqref{Eq:RHCGenSolStep1}. First, we apply the determined $\rho_{ij}^*$ value in \eqref{Eq:OP3FormalArgMinSolution} to get the optimal control decisions: $\tau_j^*$ and $\bar{\tau}_j^*$ of the control vector $U_{iaj}^*$ \eqref{Eq:RHCGenSolStep1}. 

\begin{remark}
Note that $\tau_j^*$ and $\bar{\tau}_j^*$ in \eqref{Eq:OP3FormalArgMinSolution} are piece-wise functions of $\rho_{ij}$ (with at most five cases).
Hence, $J_{sH}^*(\rho_{ij})$ in \eqref{Eq:OP3FormalArgMinSolution} is also a piece-wise function of $\rho_{ij}$. 
Even though this presents a complication to the proposed RHCP \eqref{Eq:RHCGenSolStep1} solution process, it can be resolved by considering one case (of $J_{sH}^*(\rho_{ij})$) at a time when the corresponding OCP \eqref{Eq:OCP} is solved. Then, the resulting optimal transit-time value $\rho_{ij}^*$ can be used to ensure the validity as well as the optimality of the considered case of $J_{sH}^*(\rho_{ij})$ (compared to other cases).
\end{remark}

Among the remaining control decisions in $U_{iaj}^*$ \eqref{Eq:RHCGenSolStep1}, we have already found the optimal tangential acceleration profile segment $\{u_a^*(t):t\in\mathcal{P}_{ij}\}$. Integrating this, the corresponding tangential velocity profile segment can be obtained as 
\begin{equation}\label{Eq:RHCP3OCPSolutionv}
    v_a^*(t) = \frac{6y_{ij}}{(\rho_{ij}^*)^3}(t-t_o)(t_o+\rho_{ij}^*-t),\ \forall t\in\mathcal{P}_{ij}.
\end{equation}
Finally, the optimal angular velocity profile segment  $\{w_a^*(t):t\in\mathcal{P}_{ij}\}$ (required in $U_{iaj}^*$) can be found using \eqref{Eq:RHCP3OCPSolutionv} in \eqref{Eq:AngularVelocity} together with the information about the shape of the trajectory segment $(i,j)$.

\begin{remark}
Note that the OCP \eqref{Eq:ORHCP3AgentTrajectoryOptimization} (or \eqref{Eq:OCP} in general) only requires the total length $y_{ij}$ value of the trajectory segment $(i,j)$. The shape of $(i,j)$ becomes important only when $w_a^*(t)$ has to be determined to facilitate the agent's departure from target $i$ to reach target $j$ (i.e., at the end of an \textbf{RHCP3} solving process). Therefore, even though we initially assumed the shapes of trajectory segments as prespecified, the proposed RHC framework can adapt even if they change occasionally. For instance, a new class of external events (similar to $C_j$ and $\bar{C}_j$) can be defined based on such trajectory segment shape change events - to make agents react to such events. This flexibility is an advantage as the shape of a trajectory segment may have to be designed (by an upper-level trajectory planner) taking into account moving obstacles and other agents in the mission space as well as the agent's own motion and controller constraints. 
\end{remark}

\paragraph{\textbf{Solution of RHCP \eqref{Eq:RHCGenSolStep2} for} $j^*$}
We now have solved RHCP \eqref{Eq:RHCGenSolStep1} and have obtained the optimal control vector $U_{iaj}^*$ corresponding to the next-visit target $j$. Next, this process should be repeated for all the neighboring targets $j\in\mathcal{N}_i$ to get the control vectors: $\{U_{iaj}:j\in\mathcal{N}_i\}$. Finally, the optimal next-visit target $j^*$ can be found from \eqref{Eq:RHCGenSolStep2} as $j^{\ast} = {\arg\min}_{j\in\mathcal{N}_{i}} J_{H}(X_{ia}(t_s),U_{iaj}^{\ast};\ \mathsf{w}(U_{iaj}))$.

Upon solving \textbf{RHCP3}, agent $a$ departs from the target $i$ and starts following the trajectory segment $(i,j)$ while executing the obtained optimal agent controls until it arrives at the target $j^*$. According to the proposed RHC architecture, upon arrival, the agent will solve an instance of \textbf{RHCP1}.

\subsection{Solution of \textbf{RHCP1}}
We now directly consider \textbf{RHCP1} as it encompasses \textbf{RHCP2} and is the most general form of the RHCP (\eqref{Eq:RHCGenSolStep1}-\eqref{Eq:RHCGenSolStep2}) in that no active or idle time is restricted to zero. In this case, the planning horizon $\mathsf{w}(U_{iaj})$ is same as in \eqref{Eq:VariableHorizon}.


Similar to before, we next solve \eqref{Eq:RHCGenSolStep1} (via RHCP($\rho_{ij}$) \eqref{Eq:RHCP_rho_ij} and OCP \eqref{Eq:OCP}) and \eqref{Eq:RHCGenSolStep2} to obtain the solution of \textbf{RHCP1}.

\paragraph{\textbf{Solution of RHCP($\rho_{ij}$)} \eqref{Eq:RHCP_rho_ij}} 
As mentioned before, RHCP($\rho_{ij}$) corresponding to the \textbf{RHCP1} has already been solved in \cite{Welikala2020J4} to obtain: 
\begin{equation}\label{Eq:OP1FormalArgMinSolution}
\begin{alignedat}{5}
& 
(\tau_i^*,\bar{\tau}_i^*,\tau_j^*,\bar{\tau}_j^*) 
&= 
& \underset{\makebox[2.7cm]{\footnotesize $(\tau_i,\bar{\tau}_i,\tau_j,\bar{\tau}_j)\in\mathbb{U}_{s}(\rho_{ij})$}}{\arg\min} 
& \  
& J_{sH}(\tau_i,\bar{\tau}_i,\tau_j,\bar{\tau}_j;\,\rho_{ij}) & \\
& 
J_{sH}^*(\rho_{ij}) 
&= 
& J_{sH}(\tau_i^*, \bar{\tau}_i^*, \tau_j^*, \bar{\tau}_j^*;\, \rho_{ij}) & & &
\end{alignedat}
\end{equation}
where 
\begin{equation}\label{Eq:RHCP1J_sH}
    \begin{gathered}
    \begin{aligned}
    &J_{H}(\tau_{i},\bar{\tau}_i,\tau_{j},\bar{\tau}_j;\,\rho_{ij})\\
    &=\frac{\left[
\begin{split}
C_{1}\tau_{i}^{2}+C_{2}\bar{\tau}_i^{2}+C_{3}\tau_{j}^{2}+C_{4}\bar{\tau}_j^{2}+C_{5}\tau_{i}%
\bar{\tau}_i\\
+C_{6}\tau_{i}\tau_{j}+C_{7}\tau_{i}\bar{\tau}_j+C_{8}\bar{\tau}_i\tau_{j}+C_{9}\bar{\tau}_i\bar{\tau}_j%
+C_{10}\tau_{j}\bar{\tau}_j\\
+C_{11}\tau_{i}+C_{12}\bar{\tau}_i+C_{13}\tau_{j}+C_{14}\bar{\tau}_j+C_{15}%
\end{split}
\right]}{\tau_{i}+\bar{\tau}_i+\rho_{ij}+\tau_{j}+\bar{\tau}_j},    
    \end{aligned}\\
C_{1} = \frac{\bar{A}-B_i}{2},\  
C_{2} = \frac{\bar{A}_i}{2},\ 
\bar{A}_i = \bar{A}-A_i,\ 
C_{3} = \frac{\bar{A}-B_j}{2},\\  
C_{4} = \frac{\bar{A}_j}{2},\ 
C_{5} = C_{8} = \bar{A}_{i},\   
C_{6} = \bar{A}-B_i,\ 
C_{7} = \bar{A}_j-B_i,\\   
C_{9} = \bar{A}_{ij} = \bar{A}_i-A_j,\   
C_{10} = \bar{A}_j,\   
C_{11} = [\bar{R}(t_O)+(\bar{A}-B_i)\rho_{ij}],\\  
C_{12} = [\bar{R}_i(t_O)+\bar{A}_i\rho_{ij}],\ 
\bar{R}_i = \bar{R}-R_i,\ 
C_{13} = [\bar{R}(t_O)+\bar{A}\rho_{ij}],\\ 
C_{14} = [\bar{R}_j(t_O)+\bar{A}_j\rho_{ij}],\ 
C_{15} = \frac{\rho_{ij}}{2}[2\bar{R}(t_O) + \bar{A}\rho_{ij}].        
    \end{gathered}
\end{equation}
Explicit expressions of $\tau_i^*,\bar{\tau}_i^*,\tau_j^*$ and $\bar{\tau}_j^*$ (each as a function of $\rho_{ij}$) are determined using the rational function optimization technique proposed in \cite[App.~A]{Welikala2020J4}. However, due to space constraints, we omit giving their exact forms.

\paragraph{\textbf{Objective Function of OCP \eqref{Eq:OCP}}}
The sensing objective component of OCP: $J_{sH}^*(t_f-t_o)$ now can be expressed explicitly using the $J_{sH}^*(\rho_{ij})$ function in \eqref{Eq:OP1FormalArgMinSolution} and the relationship $\rho_{ij} = t_f-t_o$. Note however that in \textbf{RHCP1}, both $t_o$ and $t_f$ are free.   Therefore, let us denote this sensing objective component of OCP as 
\begin{equation}\label{Eq:ORHCP1SensingComponent}
\phi(t_o,t_f) \triangleq J_{sH}^*(t_f-t_o).    
\end{equation}
However, the energy objective component of OCP: $J_{eH}(\{u_a(t)\})$ in \textbf{RHCP1} takes the same form as in \eqref{Eq:ORHCP3EnergyComponent}.

\paragraph{\textbf{Solution of OCP \eqref{Eq:OCP}}}
Similar to before, using  $\dot{x} = Ax(t) + Bu(t)$ with \eqref{Eq:AgentDynamicsCoeffs} to represent the agent dynamics \eqref{Eq:AgentDynamics2}, the OCP corresponding to the \textbf{RHCP1} can be stated as \begin{equation}\label{Eq:ORHCP1AgentTrajectoryOptimization}
\begin{alignedat}{3}
& \underset{\makebox[2cm]{\footnotesize $t_o,\,t_f,\,\{u(t)\}$}}{\min}   & \quad &  \alpha_H\int_{t_o}^{t_f}u^2(t)dt + \phi(t_o,t_f)\\
& \makebox[2cm]{subject to} &       & \dot{x} = Ax(t) + Bu(t),\\
&  &  & x(t_o) = [0, 0]^T,\ \ x(t_f) = [y_{ij}, 0]^T.
\end{alignedat}
\end{equation}
Note that \eqref{Eq:ORHCP1AgentTrajectoryOptimization} is a standard free initial and final time, fixed initial and final state optimal control problem.  Hence, similar to \eqref{Eq:ORHCP3AgentTrajectoryOptimization}, there is an established solution procedure \cite{bryson1975} as outlined next.

First, note that Hamiltonian corresponding to \eqref{Eq:ORHCP1AgentTrajectoryOptimization} takes the same form as in \eqref{Eq:RHCP3OCPHamiltonian}. Next, the adjoined function that combines the terminal constraints on $x(t_o)$ and $x(t_f)$ with the terminal cost $\phi(t_o,t_f)$ are written as 
\begin{equation}
\begin{aligned}
  \Phi(x(t_o),t_o,x(t_f),t_f)\ \triangleq \  &\phi(t_o,t_f)\\
  &+\nu_o^T x(t_o) + \nu_f^T(x(t_f)-[y_{ij},0]^T)
\end{aligned}
\end{equation}
where $\nu_o$ and $\nu_f$ are sets of multipliers.

Finally, the OCP in \eqref{Eq:ORHCP1AgentTrajectoryOptimization} can be solved to obtain the corresponding optimal  
$\{x(t),u(t),\lambda(t): t\in[t_o,t_f]\}, t_o, t_f, \nu_o$ and $\nu_f$ values by solving the following system of equations \cite{bryson1975}:
\begin{gather}
\label{Eq:OP1OptimalControlEquations3}
\frac{\partial H}{\partial u} = 2\alpha_Hu(t) + \lambda^T(t)B = 0, \\ \label{Eq:OP1OptimalControlEquations4}
\dot{\lambda} = -\left(\frac{\partial H}{\partial x}\right)^T = -A^T\lambda(t),\ 
\lambda(t_o) = -\left(\frac{\partial \Phi}{\partial x(t_o)}\right)^T = -\nu_o,\\  \label{Eq:OP1OptimalControlEquations5}
\lambda(t_f) = \left(\frac{\partial \Phi}{\partial x(t_f)}\right)^T = \nu_f,\\ \label{Eq:OP1OptimalControlEquations6}
\frac{\partial\Phi}{\partial t_o} - H\vert_{t=t_o} = \frac{\partial \phi}{\partial t_o} - \alpha_H u^2(t_o) + \nu_o^T(Ax(t_o)+Bu(t_o)) 
=0,\\ \label{Eq:OP1OptimalControlEquations7}
\frac{\partial\Phi}{\partial t_f} + H\vert_{t=t_f} =  
\frac{\partial \phi}{\partial t_f} + \alpha_H u^2(t_f) + \nu_f^T(Ax(t_f)+Bu(t_f))
= 0 
\end{gather}
in addition to the agent dynamics and terminal constraints given in \eqref{Eq:ORHCP1AgentTrajectoryOptimization}.
Note that \eqref{Eq:OP1OptimalControlEquations3} is the optimality condition (from Pontryagin’s minimum principle), \eqref{Eq:OP1OptimalControlEquations4}-\eqref{Eq:OP1OptimalControlEquations5} are the co-state equations and \eqref{Eq:OP1OptimalControlEquations6}-\eqref{Eq:OP1OptimalControlEquations7} are the transversality conditions.

\begin{lemma}\label{Lm:RHCP1OCPSolution}
The optimal transit time $\rho_{ij}^* = t_f^*-t_o^*$ of the OCP \eqref{Eq:ORHCP1AgentTrajectoryOptimization} satisfies the equation
\begin{eqnarray}
\label{Eq:OptimalTransitTimeEquationRHCP1}
\rho_{ij}^4 \, \frac{d\phi(t_o,t_f)}{d \rho_{ij}} = 36 \alpha y_{ij}^2,  
\end{eqnarray}
where $\phi(t_o,t_f)$ \eqref{Eq:ORHCP1SensingComponent} is considered as a function of $\rho_{ij}=t_f-t_o$. Thus, the optimal terminal times $t_o^*$ and $t_f^*$ of \eqref{Eq:ORHCP1AgentTrajectoryOptimization} are  
\begin{equation}\label{Eq:RHCP1TerminalTimes}
    t_o^* = t_s + \tau_i^*(\rho_{ij}^*) + \bar{\tau}_i^*(\rho_{ij}^*)\ \ \mbox{ and }\ \ 
    t_f^* = t_o + \rho_{ij}^*,
\end{equation}
respectively ($\tau_i^*(\rho_{ij})$ and $\bar{\tau}_i^*(\rho_{ij})$ are as in  \eqref{Eq:OP1FormalArgMinSolution}). The corresponding optimal control law $u^*(t)$ of \eqref{Eq:ORHCP1AgentTrajectoryOptimization} is given by 
\begin{equation}\label{Eq:RHCP1ControlSolution}
    u^*(t) = \frac{12 y_{ij}}{(t_f^*-t_o^*)
    ^3}\left[\frac{t_f^*+t_o^*}{2} - t\right] 
    ,\ \ \forall t\in[t_o,t_f].
\end{equation}
\end{lemma}
\begin{proof}
The proof follows the same steps as that of Lemma \ref{Lm:RHCP3OCPSolution} and is, therefore, omitted. However, the main steps are: (i) solve for $\lambda(t), u(t)$ and $x(t), \forall t\in[t_o,t_f]$ using \eqref{Eq:OP1OptimalControlEquations4}, \eqref{Eq:OP1OptimalControlEquations3} and agent dynamics, respectively, in that order, in terms of $t_o,t_f,\nu_o$ and $\nu_f$.
(ii) use the terminal constraint $x(t_f)=[y_{ij},0]^T$ and \eqref{Eq:OP1OptimalControlEquations4} to determine $\nu_o,\nu_f$ in terms of $t_o,t_f$ and (iii) solve for $t_o$ and $t_f$ using \eqref{Eq:OP1OptimalControlEquations5} and \eqref{Eq:OP1OptimalControlEquations6}.    
\end{proof}

Using the optimal terminal times $t_f^*,t_o^*$ and the control $u^*(t)$ (i.e., $u_a^*(t)$) proven in Lemma \ref{Lm:RHCP1OCPSolution}, the optimal energy objective component of this OCP \eqref{Eq:ORHCP1AgentTrajectoryOptimization} can be evaluated as
\begin{equation}\label{Eq:RHCP1EnergySolution}
    J_{eH}(\{u^*(t)\}) = \frac{12y_{ij}^2}{(t_f^*-t_o^*)^3}. 
\end{equation}
The corresponding sensing objective component \eqref{Eq:ORHCP1SensingComponent} is directly given by $\phi(t_o^*,t_f^*) = J_{sH}^*(t_f^*-t_o^*) = J_{sH}(\rho_{ij}^*)$.

\paragraph{\textbf{Solution of RHCP \eqref{Eq:RHCGenSolStep1} for} $U_{iaj}^*$}
We now conclude the solution process of RHCP \eqref{Eq:RHCGenSolStep1} (outlined in Fig. \ref{Fig:OverviewOfTheRHCPSolution}) by applying the determined optimal transit-time $\rho_{ij}^*$ in \eqref{Eq:OP1FormalArgMinSolution} to obtain the optimal control decisions $\tau_i^*, \bar{\tau}_i^*, \tau_j^*$ and $\bar{\tau}_j^*$ included in the control vector $U_{iaj}^*$ \eqref{Eq:RHCGenSolStep1}. Note that here it is not necessary to evaluate the optimal agent angular velocity profile $\{w_a^*(t):t\in[t_o^*,t_f^*]\}$ (unlike in \textbf{RHCP3}) as the agent does not plan to depart from the current target immediately. 

\paragraph{\textbf{Solution of RHCP \eqref{Eq:RHCGenSolStep2} for} $j^*$}
We now have solved the RHCP \eqref{Eq:RHCGenSolStep1} and have obtained the optimal control vector $U_{iaj}^*$ corresponding to a next-visit target $j\in\mathcal{N}_i$. Executing this process for all $j\in\mathcal{N}_i$ and subsequently evaluating \eqref{Eq:RHCGenSolStep2} gives the optimal next-visit target $j^*$ as $j^{\ast} = {\arg\min}_{j\in\mathcal{N}_{i}} J_{H}(X_{ia}(t_s),U_{iaj}^*;\ \mathsf{w}(U_{iaj}^*))$.

Upon solving \textbf{RHCP1}, agent $a$ remains active on target $i$ for a duration of $\tau_i^*$, and in the meantime, if any other external event such as $C_j$ or $\bar{C}_j$ for some $j\in\mathcal{N}_i$ occurs, it re-computes the remaining active time at target $i$. However, if the agent completes executing the  determined active time (i.e., if the corresponding event $[h\rightarrow\tau_i^*]$ occurs) with $R_i(t_O+\tau_i^*)=0$, then, the agent will have to subsequently solve an instance of \textbf{RHCP2} to determine the remaining inactive time at target $i$. Otherwise (i.e.,  if the event $[h\rightarrow\tau_i^*]$ occurs with $R_i(t_O+\tau_i^*)>0$), the agent will have to subsequently solve an instance of \textbf{RHCP3} to determine the next-visit target and depart from target $i$.

\begin{remark}
Upon solving \textbf{RHCP1}, over the subsequent active time at target $i$, the agent can choose to control its angular velocity $w_a(t)$ (while keeping $u_a(t)=0$) 
to adjust its heading $\theta_a(t)$ to accommodate the impending departure towards the next-visit target $j^*$ found in \eqref{Eq:RHCGenSolStep2}.
However, the agent can also choose to keep $w_a(t)=0$ over such dwell-time periods and plan the shapes of the trajectory segments $\{(i,j):j\in\mathcal{N}_i\}$ accordingly. 
\end{remark}

\section{Optimal Controls for First-Order Agents}
\label{Sec:FOMethods}

In the previous sections, we have proposed an RHC based solution for PMN problems that uses energy-aware second-order agents. In this section, for comparison purposes, we first present the details of a similar RHC solution \cite{Welikala2020J4} that uses energy-agnostic first-order agents. Subsequently, motivated by a few practical qualities that such first-order agent behaviors (controls) possess, we derive energy-aware control laws for governing actual second-order agents in a way that they imitate first-order agents.

In particular, this section explores how the agent controls $\{u_a^*(t):t\in\mathcal{P}_{ij}\}$ derived for energy-aware second-order agents \eqref{Eq:AgentDynamics} (given in Lemmas \ref{Lm:RHCP3OCPSolution}, \ref{Lm:RHCP1OCPSolution}) should be modified if we are to replace them with energy-agnostic first-order agents or with energy-aware second-order agents that imitate first-order agents. Based on the proposed modular RHCP \eqref{Eq:RHCGenSolStep1} solution process (outlined in Fig. \ref{Fig:OverviewOfTheRHCPSolution}), note that, a change in the agent dynamic model will affect only the OCP \eqref{Eq:OCP} (not the RHCP($\rho_{ij}$) \eqref{Eq:RHCP_rho_ij}) - specifically through the energy objective component $J_{eH}(\{u_a(t)\})$. Therefore, in this section, our main focus is on re-stating (and solving) the OCP \eqref{Eq:OCP} assuming its sensing objective component $J_{sH}^*(\rho_{ij})$ is given. Note also that such a change in the OCP objective will directly affect the resulting optimal transit-time (i.e., $\rho_{ij}^*$) and consequently will also affect all the other control variables in $U_{iaj}^*$ \eqref{Eq:RHCGenSolStep1}. 

Since we assume the $J_{sH}^*(\rho_{ij})$ function is given in this section, we keep the ensuing analysis independent of the exact RHCP form (i.e., \textbf{RHCP1}, \textbf{RHCP2} or \textbf{RHCP3}). To this end, we start with generalizing the optimal agent controls established for second-order agents in Lemmas \ref{Lm:RHCP3OCPSolution} and \ref{Lm:RHCP1OCPSolution}, in a theorem. For convenience, let us label the PMN solution that uses energy-aware actual second-order agents (developed in the previous sections) as the ``SO Method''.

\paragraph{\textbf{SO Method}} 
First, note that both equations \eqref{Eq:RHCP3OCPSolutiont_f} and \eqref{Eq:OptimalTransitTimeEquationRHCP1} are equivalent and thus can be written generally as 
\begin{equation}\label{Eq:OptimumTransitTimeGeneral}
    \rho_{ij}^4 \, \frac{J_{sH}^*(\rho_{ij})}{d\rho_{ij}} = 36 \alpha y_{ij}^2.
\end{equation}
The $\rho_{ij}$ value that satisfies the above equation is the optimal transit-time to be used under the SO method (irrespective of the RHCP form). Second, note that both \eqref{Eq:RHCP3EnergySolution} and \eqref{Eq:RHCP1EnergySolution} are also equivalent. This implies that the optimal agent energy consumption can be expressed independent of the RHCP form. Finally, note that both \eqref{Eq:RHCP3OCPSolutionu_t} and \eqref{Eq:RHCP1ControlSolution} represent the same agent control profile albeit with a shifted starting point $t_o=t_o^*$ (compared to being $t_o=t_s$) in the latter. However, as shown in \eqref{Eq:RHCP1TerminalTimes}, this starting point $t_o^*$ depends only on sensing objective related quantities and the determined optimal transit-time value. Moreover, since our main focus in this section is on agent controls $\{u_a(t):t\in\mathcal{P}_{ij}\}$, without loss of generality, we can assume $t_o=t_o^*=0$ and $\mathcal{P}_{ij} = [0,\rho_{ij}]$. 


With regard to the SO method, let us denote: 
(i) optimal transit-time as $\rho_{SO}$,\ \ 
(ii) optimal tangential velocity and acceleration as $v_a^*(t)$ and $u_a^*(t)$, respectively, for $t\in[0,\rho_{SO}]$,\ \ 
(iii) maximum tangential velocity and acceleration as $v_{SO} \triangleq \max\{v_a^*(t)\}$ and $u_{SO}\triangleq \max\{u_a^*(t)\}$, respectively, and, 
(iv) optimal agent energy consumption for the transition as $E_{SO}$.

\begin{theorem}\label{Th:SOMethodOptimalControls}
Under the SO method: $\rho_{SO}$ is given by \eqref{Eq:OptimumTransitTimeGeneral}, 
\begin{equation}\label{Eq:SOMethodProfiles}
    v_a^*(t) = \frac{6y_{ij}}{\rho_{SO}^3}\,t\,(\rho_{SO}-t),     \ \ \ \ 
    u_a^*(t) = \frac{6y_{ij}}{\rho_{SO}^2}\left(1-\frac{2t}{\rho_{SO}}\right),
\end{equation}
for $t\in[0,\rho_{SO}]$ and 
\begin{equation}\label{Eq:SOMethodMaxValues}
    v_{SO} = \frac{3y_{ij}}{2\rho_{SO}},\ \ 
    u_{SO} = \frac{6y_{ij}}{\rho_{SO}^2},\ \ 
    E_{SO} = \frac{12 y_{ij}^2}{\rho_{SO}^3}.
\end{equation}
\end{theorem}
\begin{proof}
The results stated in \eqref{Eq:SOMethodProfiles} directly follow from \eqref{Eq:RHCP3OCPSolutionu_t} and \eqref{Eq:RHCP3OCPSolutionv}. The relationships given in \eqref{Eq:SOMethodMaxValues} can be obtained using \eqref{Eq:SOMethodProfiles} (via calculus) and \eqref{Eq:RHCP3EnergySolution}.
\end{proof}

Figure \ref{Fig:FirstOrderAgentStateTrajectory} illustrates an example agent tangential velocity profile segment $\{v_a^*(t):t\in[0,\rho_{SO}]\}$. In the subsequent subsections, we explore several alternative approaches to this SO method. In particular, each of these alternative methods has its root in the first-order agent model used in \cite{Welikala2020J4,Zhou2019} - where each agent is assumed to travel at a fixed predefined velocity over each trajectory segment, and thus, does not involve an OCP in RHCPs. We label this approach of controlling agents as the ``FO-0 Method'' and an example agent tangential velocity profile observed is shown in Fig. \ref{Fig:FirstOrderAgentStateTrajectory}. 

\begin{figure}[!t]
    \centering
    \includegraphics[width=2.7in]{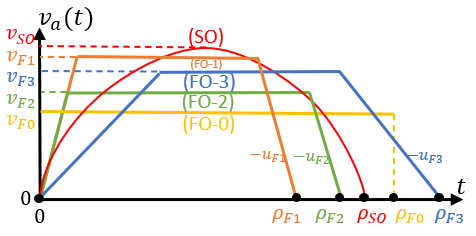}
    \caption{Tangential velocity profiles on a trajectory segment $(i,j) \in \mathcal{E}$ of length $y_{ij}$ under second-order (SO) and four different forms of (approximate) first-order (FO-0,1,2,3) agent models.
    Under the FO-$n$ agent model $n\in\{0,1,2,3\}$, $\rho_{Fn} \triangleq y_{ij}/v_{mn}$ where $v_{mn}$ is the average velocity, $v_{Fn}$ is the maximum velocity and $u_{Fn}$ is the maximum absolute acceleration level. 
    }
    \label{Fig:FirstOrderAgentStateTrajectory}
\end{figure}


However, note that we can neither characterize the total energy consumption nor control a real-world agent over such a tangential velocity profile as in the FO-0 curve in Fig. \ref{Fig:FirstOrderAgentStateTrajectory} - due to the involved instantaneous infinite accelerations. 
Therefore, to facilitate a comparison between SO and FO-0 methods, we propose to use actual second-order agents (instead of first-order ones) but enforce each agent controller to approximate a first-order agent behavior (FO-0). We label this approximate version of the FO-0 method as the ``FO-1 Method'' and a corresponding agent tangential velocity profile is shown in Fig. \ref{Fig:FirstOrderAgentStateTrajectory}.



\paragraph{\textbf{FO-1 Method}} 
Under the FO-1 method, as shown in Fig. \ref{Fig:FirstOrderAgentStateTrajectory}, each agent is assumed to go through a sequence of: constant acceleration (of $u_{F1}$), constant velocity (of $v_{F1}$) and constant deceleration (of $-u_{F1}$) stages over a period of length $\rho_{F1}$ every time it travels on a trajectory segment. In particular, the acceleration/deceleration magnitude $u_{F1}$ and the average velocity value $v_{m1} = y_{ij}/\rho_{F1}$ are assumed to be prespecified, commonly for all $(i,j)\in\mathcal{E}$. The resulting maximum velocity level on a trajectory segment $(i,j)\in\mathcal{E}$ is denoted as $v_{F1}^{ij}$ and can be expressed in terms of $y_{ij},u_{F1}$ and $v_{m1}$ as  
\begin{equation}\label{Eq:FO1MethodMaxEdgeVelocity}
    v_{F1}^{ij}(u_{F1},v_{m1}) =
    \begin{cases}
    \frac{y_{ij}u_{F1}-\sqrt{y_{ij}^2u_{F1}^2-4v_{m1}^2y_{ij}u_{F1}}}{2v_{m1}}&\ \mbox{ if } y_{ij}\geq \frac{4v_{m1}^2}{u_{F1}}\\
    \sqrt{y_{ij}u_{F1}}&\ \mbox{ otherwise. }
    \end{cases}
\end{equation}

To conduct a fair comparison between the SO and FO-1 methods, the two parameters $u_{F1}$ and $v_{m1}$ that define the FO-1 method are selected as follows. First, let us define $u_{SO}^{max}$ and $v_{SO}^{max}$ as the respective maximum values of all the $u_{SO}$ and $v_{SO}$ values (empirical) observed in the interested PMN problem. Then, we propose to enforce: \begin{equation}\label{Eq:FO1MethodParameters}
\begin{aligned}
  &u_{F1} =\ u_{SO}^{max} \ \mbox{ and }\\
  & \begin{alignedat}{4}
& v_{m1} = & \underset{\makebox[2.5cm]{\footnotesize $v_m > 0,\ (i,j)\in\mathcal{E}$}}{\makebox[2.5cm]{argmax}}   & \quad &&  v_{F1}^{ij}(u_{SO}^{max},v_m)\\
& & \makebox[2.5cm]{subject to} &       && v_{F1}^{ij}(u_{SO}^{max},v_m)\leq v_{SO}^{max},
\end{alignedat}
\end{aligned}
\end{equation}
to ensure the maximum velocity and acceleration values resulting from the FO-1 method are identical or as close as possible to those of the SO method. 

\begin{proposition}\label{Pr:FO1}
The $v_{m1}$ expression given in \eqref{Eq:FO1MethodParameters} can be simplified into the form
\begin{equation}\label{Eq:FO1Parameters2}
\begin{alignedat}{4}
& v_{m1} = & \underset{\makebox[2cm]{\footnotesize $(i,j)\in\mathcal{E}$}}{\min} & \quad && \frac{y_{ij}u_{SO}^{max}v_{SO}^{max}}{(v_{SO}^{max})^2+y_{ij}u_{SO}^{max}} & \\
&  & \makebox[2cm]{subject to} &  && y_{ij}\geq  (v_{SO}^{max})^2/u_{SO}^{max}. 
\end{alignedat}
\end{equation}
\end{proposition}
\begin{proof}
Provided in Appendix \ref{App:Lm_FO1_Proof}.
\end{proof}
Note that according to Proposition \ref{Pr:FO1}, we need to assume the given $u_{SO}^{max}$ and $v_{SO}^{max}$ satisfy that $\exists (i,j)\in\mathcal{E}$ with $y_{ij}\geq  (v_{SO}^{max})^2/u_{SO}^{max}$ $v_{m1}$. However, if this assumption does not hold, we simply can use a lower $v_{SO}^{max}$ value than its actual value when evaluating $v_{m1}$ in \eqref{Eq:FO1Parameters2}. Note also that the maximum velocity value observed in the FO-1 method is given by 
\begin{equation}
    v_{F1} = \max_{(i,j)\in\mathcal{E}}\ v_{F1}^{ij}(u_{SO}^{max},v_{m1}),
\end{equation}
and the agent energy consumption on a trajectory segment $(i,j)$ can be proven to be
\begin{equation}
  E_{F1}^{ij} = 2u_{SO}^{max}v_{F1}^{ij}(u_{SO}^{max},v_{m1}).  
\end{equation}

We can now compare the FO-1 and SO methods as we can compute the total agent energy consumption in the FO-1 method (numerical results are provided in Section \ref{Sec:ExtentionsAndExamples}, e.g., Tab \ref{Tab:Results}). We again highlight that the FO-1 method: (i) does not consider the agent energy when solving its RHCPs (i.e., RHCPs do not involve an OCP) and (ii) uses actual second-order agents whose controllers constrained to approximate first-order agent behaviors. We conclude our discussion about the FO-1 method with the following remark - which will motivate us to refine the proposed FO-1 method.

\begin{remark}
Notice that the optimal second-order agent control $u_a^*(t)$ \eqref{Eq:SOMethodProfiles} (in the SO method) decreases linearly and includes a zero-crossing point (at $t=(t_o+t_f)/2$). In practice, such a control input can be difficult to realize due to dead-bands in response of the used motion actuator (near $u_a(t) = 0$). In contrast, the bang-zero-bang type of a control input required when controlling a second-order agent so that it approximates a first-order agent (like in the FO-1 method) - can be conveniently implemented.  
\end{remark}

\paragraph{\textbf{FO-2 Method}} 
Even though we have proposed a reasonable and consistent way to select the parameters involved in the FO-1 method (i.e., $u_{F1}$ and $v_{m1}$), it is clear that such an approach is agnostic to the agent energy consumption. To address this concern, we next propose the FO-2 method, which as shown in Fig. \ref{Fig:FirstOrderAgentStateTrajectory}, is identical to the FO-1 method in many ways except for its choice of acceleration/deceleration magnitude $u_{F2}$ and average velocity value $v_{m2}$. In particular, as opposed to selecting $u_{F2},v_{m2}$ according to \eqref{Eq:FO1MethodParameters}, here, an energy-optimized approach is followed.


\begin{theorem}\label{Th:FO2MethodOptimalControls}
Under the FO-2 method, for a fixed average velocity $v_{m2} = y_{ij}/\rho_{F2}$, on a trajectory segment $(i,j)\in\mathcal{E}$, the optimal agent energy consumption is $E_{F2} = \frac{27v_{m2}^3}{2y_{ij}}$ and it is achieved when $u_{F2}=\frac{9v_{m2}^2}{2y_{ij}}$ and $v_{F2}= \frac{3v_{m2}}{2}$ are used.
\end{theorem}
\begin{proof}
Since the total distance traveled by the (FO-2) agent over the period $[0,\rho_{F2}]$ is $y_{ij}$, we can state that
\begin{equation}\label{Eq:FO2MethodProofStep1}
    \frac{1}{2}\left(\rho_{F2} + \left(\rho_{F2}-2\frac{v_{F2}}{u_{F2}}\right)\right)v_{F2} = y_{ij}.
\end{equation}
Over the same period, the corresponding total agent energy requirement (denoted as $E_{F2}$) can be evaluated by integrating the square of the acceleration profile used. This gives
\begin{equation}\label{Eq:FO2MethodProofStep2}
    E_{F2} = \int_0^{\rho_{F2}}u^2(t)dt = 2u_{F2}^2 \frac{v_{F2}}{u_{F2}} = 2u_{F2}v_{F2}. 
\end{equation}
This expression can be further simplified using \eqref{Eq:FO2MethodProofStep1} to obtain
\begin{equation}\label{Eq:FO2MethodProofStep3}
    E_{F2} = \frac{2v_{F2}^3}{v_{F2}\rho_{F2}-y_{ij}}.
\end{equation}
Recall that both $y_{ij}$ and $\rho_{F2}(=y_{ij}/v_{m2})$ are fixed in this case. Therefore, $E_{F2}$ in \eqref{Eq:FO2MethodProofStep3} is a function of (only) $v_{F2}$. Thus, we can use calculus to determine the choice of $v_{F2}$ that minimizes $E_{F2}$. This (and back substitution) reveals:
\begin{equation}\label{Eq:FO2MethodProofStep4}
\begin{gathered}
v_{F2} = \frac{3y_{ij}}{2\rho_{F2}},\ \ 
u_{F2} = \frac{9y_{ij}}{2\rho_{F2}^2},\ \ 
E_{F2} = \frac{27y_{ij}^2}{2\rho_{F2}^3}.
\end{gathered}
\end{equation}
Finally, this proof can be completed by replacing $\rho_{F2}$ terms with $y_{ij}/v_{m2}$ in each of the above expressions.
\end{proof}

\begin{corollary}\label{CO:FO2}
If $v_{m2}$ in the FO-2 method is such that $\frac{y_{ij}}{v_{m2}} = \rho_{SO}$ (i.e., $\rho_{F2} = \rho_{SO}$), then $$v_{F2}=v_{SO},\ \ u_{F2}=\frac{3}{4}u_{SO},\ \ E_{F2} = \frac{9}{8}E_{SO}.$$
\end{corollary}
\begin{proof}
This result directly follow from comparing Theorem \ref{Th:SOMethodOptimalControls} \eqref{Eq:SOMethodMaxValues} with \eqref{Eq:FO2MethodProofStep4}. 
\end{proof}

Next, we use Theorem \ref{Th:FO2MethodOptimalControls} to develop energy-optimized choices for the $v_{m2}$ and $u_{F2}$ parameters of the FO-2 method. However, similar to before, we also use $v_{SO}^{max}$ and $u_{SO}^{max}$ values (empirical) as known inputs in this process to make sure the maximum velocity and acceleration values resulting from the FO-2 method are identical or as close as possible to those of the SO method.

Note that the optimal choices of $v_{F2}$ and $u_{F2}$ given in Theorem \ref{Th:FO2MethodOptimalControls} are dependent on both $(i,j)\in\mathcal{E}$ and $v_{m2}$. Therefore, let us denote those as functions:
\begin{equation}\label{Eq:FO2Parameter0}
v_{F2}^{ij}(v_{m2}) = \frac{3}{2}v_{m2},\ \ 
u_{F2}^{ij}(v_{m2}) = \frac{9v_{m2}^2}{2y_{ij}}.
\end{equation}
Now, we propose to select the parameter $v_{m2}$ based on the above two relationships and the given $v_{SO}^{max}$, $u_{SO}^{max}$ values as
\begin{equation}\label{Eq:FO2Parameter}
\begin{alignedat}{4}
& v_{m2} = & \underset{\makebox[2.5cm]{\footnotesize $v_m>0,\ (i,j)\in\mathcal{E}$}}{\arg\max}
& \quad && v_{F2}^{ij}(v_{m}) & \\
& & \makebox[2.5cm]{subject to} & && v_{F2}^{ij}(v_{m})\leq v_{SO}^{max}, &\\
& & & && u_{F2}^{ij}(v_{m})\leq u_{SO}^{max}. &
\end{alignedat}
\end{equation}

\begin{proposition}\label{Pr:FO2}
The $v_{m2}$ expression given in \eqref{Eq:FO2Parameter} can be simplified into the form
\begin{equation}
    v_{m2} = \min\ \{\frac{1}{3}\sqrt{2y_{min} u_{SO}^{max}},\ \frac{2}{3}v_{SO}^{max}\}
\end{equation}
where $y_{min} = \min_{(i,j)\in\mathcal{E}}y_{ij}$.
\end{proposition}
\begin{proof}
The proof follows the same steps as that of Proposition \ref{Pr:FO1} and is, therefore, omitted.
\end{proof}

We point out that even though the average velocity $v_{m2}$ computed above is used commonly across all the trajectory segments, the acceleration/deceleration level of an agent in a trajectory segment $(i,j)$ has to be selected as $u_{F2}^{ij}(v_{m2})$ \eqref{Eq:FO2Parameter0} so as to \emph{optimize} the agent energy consumption. Hence, the overall maximum acceleration/deceleration level observed in the FO-2 method is (via \eqref{Eq:FO2Parameter0})
\begin{equation}
u_{F2} = \max_{(i,j)\in\mathcal{E}}\  u_{F2}^{ij}(v_{m2}).  
\end{equation}

Consider the scenario where $\rho_{F2}=\rho_{SO}$ on a certain trajectory segment. In such a case, sensing objective wise, both FO-2 and SO methods perform equally. However, Corollary \ref{CO:FO2} states that energy objective wise, the FO-2 method shows a $12.5\%$ loss (i.e., a higher energy consumption) compared to the SO method. Moreover, recall that the FO-2 method does not consider the energy expenditure when solving its RHCPs (i.e., no OCP is involved, similar to FO-0 and FO-1 methods). To mitigate these two obvious disadvantages, we next propose the FO-3 method - where we try to optimize the energy objective further compromising the sensing objective in an OCP.

\paragraph{\textbf{FO-3 Method}}

As shown in Fig. \ref{Fig:FirstOrderAgentStateTrajectory}, the FO-3 method has similarly shaped agent state trajectories to FO-1 and FO-2 methods. However, as we will see next, the FO-3 method does not involve any parameter that needs to be selected based on external information like $u_{SO}^{max}$ and $v_{SO}^{max}$. 

On the other hand, note that in the FO-2 method, the optimal agent energy consumption  $E_{F2}$ \eqref{Eq:FO2MethodProofStep4} is inversely proportional to the transit-time $\rho_{F2}$. Motivated by this, the FO-3 method proposes to use a larger transit-time $\rho_{F3} \geq \rho_{F2}$ (see Fig. \ref{Fig:FirstOrderAgentStateTrajectory}) compromising the sensing objective so as to achieve a better (lower) energy objective. However, to make this trade-off a profitable one (in terms of the total objective \eqref{Eq:MainObjective}) we need to use the OCP \eqref{Eq:OCP}.

Note that we assume $J_{sH}^*(\rho_{ij})$ (i.e., the sensing objective component of the OCP \eqref{Eq:OCP}) as a known function in this section. Therefore, the sensing objective component of OCP \eqref{Eq:OCP} under the FO-3 method can be written as $J_{sH}^*(\rho_{F3})$. On the other hand, under the FO-3 method, the energy objective component of the OCP can be written as $J_{eH}(\{u_a(t)\}) = E_{F3} \triangleq \frac{27y_{ij}^2}{2\rho_{F3}^3}$ (using $E_{F2}$ in \eqref{Eq:FO2MethodProofStep4} and replacing $\rho_{F2}$ with $\rho_{F3}$). Hence, the objective function of the OCP \eqref{Eq:OCP} under the FO-3 method is 
\begin{equation}\label{Eq:FO3CompositeObjective}
    J_H = \alpha_H E_{F3} + J_{sH}^*(\rho_{F3}) = \alpha_H \frac{27y_{ij}^2}{2\rho_{F3}^3} + J_{sH}^*(\rho_{F3}).
\end{equation}

\begin{theorem}\label{Th:FO3MethodOptimalControls}
Under the FO-3 method, the optimal transit-time is $\rho_{ij} = \rho_{F3}$ that satisfies the equation:
\begin{equation}\label{Eq:FO3MethodOptimalControls1}
    \rho_{ij}^4\,\frac{dJ_{sH}^*(\rho_{ij})}{d\rho_{ij}} = \frac{81}{2}\alpha y_{ij}^2.
\end{equation}
The corresponding optimal values of $v_{F3}$, $u_{F3}$ and $E_{F3}$ are 
\begin{equation}\label{Eq:FO3MethodOptimalControls2}
v_{F3} = \frac{3y_{ij}}{2\rho_{F3}},\ \ 
u_{F3} = \frac{9y_{ij}}{2\rho_{F3}^2},\ \ 
E_{F3} = \frac{27y_{ij}^2}{2\rho_{F3}^3},
\end{equation}
i.e., 
$
v_{F3} = \frac{1}{k}v_{SO},\ 
u_{F3} = \frac{3}{4k^2}u_{SO},\ 
E_{F3} = \frac{9}{8k^3}E_{SO}
$
where $k\triangleq\frac{\rho_{F3}}{\rho_{SO}}$.
\end{theorem}

\emph{Proof: }
The OCP objective $J_H$ given in \eqref{Eq:FO3CompositeObjective} depends only on the choice of $\rho_{F3}$. Therefore, the optimal $\rho_{F3}$ value that minimizes $J_H$ can be found using the equation $\frac{dJ_H}{d\rho_{F3}}=0$, which translates into \eqref{Eq:FO3MethodOptimalControls1}. Since both FO-2 and FO-3 methods assume structurally similar velocity profiles, we still can use Theorem \ref{Th:FO2MethodOptimalControls} in the context of FO-3 after replacing $\rho_{F2}$ with $\rho_{F3}$. In this way, \eqref{Eq:FO3MethodOptimalControls2} (and the remaining results) can be obtained using \eqref{Eq:FO2MethodProofStep4} (and Theorem \ref{Th:SOMethodOptimalControls} \eqref{Eq:SOMethodMaxValues}).
\hfill $\blacksquare$

Note that even though equations \eqref{Eq:FO3MethodOptimalControls1} and \eqref{Eq:OptimumTransitTimeGeneral} are structurally similar, their subtle difference (in the coefficient on the RHS) causes the FO-3 method to have a different transit-time value compared to the SO method. Based on the difference between \eqref{Eq:FO3MethodOptimalControls1} and \eqref{Eq:OptimumTransitTimeGeneral}, we can anticipate $\rho_{F3}>\rho_{SO}$ (also, as we intended in the first place). In such a case, the parameter $k$ defined in Theorem \ref{Th:FO3MethodOptimalControls} follows $k>1$. This implies that (from Theorem \ref{Th:FO3MethodOptimalControls}) $v_{F3}<v_{SO}$ and $u_{F3}<u_{SO}$, i.e., the FO-3 method requires smaller velocity and acceleration values compared to the SO method.

As shown in Appendix \ref{SubSec:SOAgentsWithConstraints} and \ref{SubSec:FOAgentsWithConstraints}, this analysis regarding the optimal (approximate) first and second-order agent behaviors on trajectory segments can be extended effortlessly for scenarios where agents have additional velocity and acceleration constraints.

\section{Numerical Results}
\label{Sec:ExtentionsAndExamples}

In this section, we first explore the nature of individual \textbf{RHCP3} and \textbf{RHCP1} solutions presented in Section \ref{Sec:SolutionToTheRHCPs} under second-order agents (i.e., SO method). 
Then, we compare the performance metrics $J_T, J_e$ and $J_s$ defined in \eqref{Eq:MainObjective} obtained for several different PMN problem configurations (shown in Fig. \ref{Fig:FinalConditions}) under different agent control methods: SO, FO-1, FO-2 and FO-3.

\subsection{Numerical Results for a \textbf{RHCP3}}
\label{SubSec:NumResRHCP3}

To numerically evaluate quantities relating to a \textbf{RHCP3} and its solution, we chose target parameters as $A_m = 1$, $B_m = 10, \forall m\in\mathcal{N}_i$ with $\vert \bar{\mathcal{N}}_i \vert = 3$. Moreover, the \emph{default} values of $\alpha,\,R_j(t_s),\,y_{ij}$ and $H$ are chosen respectively as $0.5,\,100,\,50$ and $250$. Figures \ref{Fig:ORHCP3SolWithAlpha} - \ref{Fig:ORHCP3SolWithy_ij} respectively show how the \textbf{RHCP3} solution (i.e., $v_a(t), u_a(t), J_{sH}, J_{eH}$ and $J_H$) changes when the three parameters $\alpha, R_j(t_s)$ and $y_{ij}$ are varied.

\begin{figure}[!b]
     \centering
     \begin{subfigure}[b]{0.32\columnwidth}
         \centering
         \includegraphics[width = \textwidth]{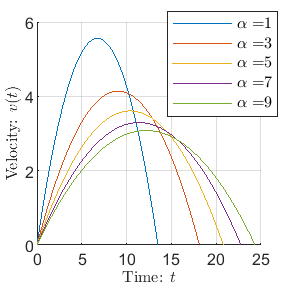}
         \caption{Tan. Vel.:$v_a(t)$}
         
     \end{subfigure} 
     \hfill
     \begin{subfigure}[b]{0.32\columnwidth}
         \centering
         \includegraphics[width = \textwidth]{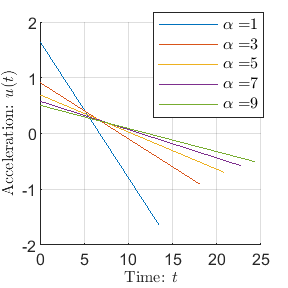}
         \caption{Tan. Acc.: $u_a(t)$}
         
     \end{subfigure} 
     \hfill
     \begin{subfigure}[b]{0.32\columnwidth}
         \centering
         \includegraphics[width = \textwidth]{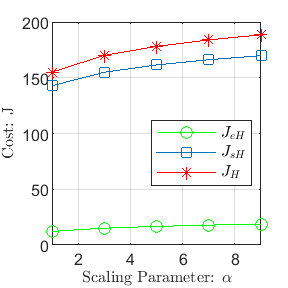}
         \caption{Cost: $J_{sH}, J_{eH}, J_{H}$}
         
     \end{subfigure}
     \caption{\textbf{RHCP3} solution under different weight factors (i.e., $\alpha$ in\eqref{Eq:MainObjective}) values.}
    \label{Fig:ORHCP3SolWithAlpha}
\end{figure}

\begin{figure}[!b]
     \centering
     \begin{subfigure}[b]{0.32\columnwidth}
         \centering
         \includegraphics[width = \textwidth]{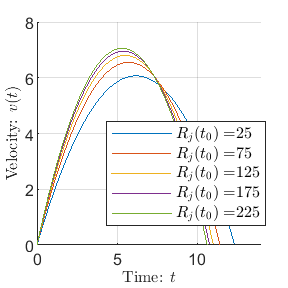}
         \caption{Tan. Vel.:$v_a(t)$}
         
     \end{subfigure} 
     \hfill
     \begin{subfigure}[b]{0.32\columnwidth}
         \centering
         \includegraphics[width = \textwidth]{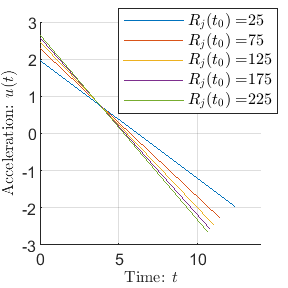}
         \caption{Tan. Acc.: $u_a(t)$}
         
     \end{subfigure} 
     \hfill
     \begin{subfigure}[b]{0.32\columnwidth}
         \centering
         \includegraphics[width = \textwidth]{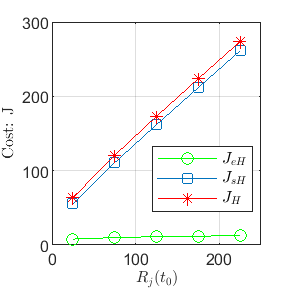}
         \caption{Cost: $J_{sH}, J_{eH}, J_{H}$}
         
     \end{subfigure}
     \caption{\textbf{RHCP3} solution under different initial target uncertainty (i.e., $R_j(t_s)$) values.}
    \label{Fig:ORHCP3SolWithR_j}
\end{figure}

\begin{figure}[!b]
     \centering
     \begin{subfigure}[b]{0.32\columnwidth}
         \centering
         \includegraphics[width = \textwidth]{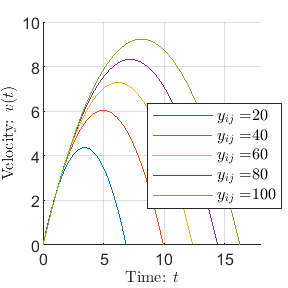}
         \caption{Tan. Vel.:$v_a(t)$}
         
     \end{subfigure} 
     \hfill
     \begin{subfigure}[b]{0.32\columnwidth}
         \centering
         \includegraphics[width = \textwidth]{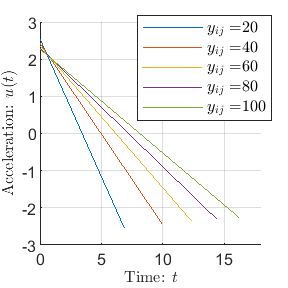}
         \caption{Tan. Acc.: $u_a(t)$}
         
     \end{subfigure} 
     \hfill
     \begin{subfigure}[b]{0.32\columnwidth}
         \centering
         \includegraphics[width = \textwidth]{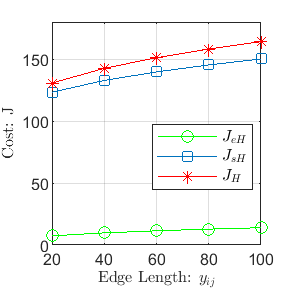}
         \caption{Cost: $J_{sH}, J_{eH}, J_{H}$}
         
     \end{subfigure}
     \caption{\textbf{RHCP3} solution under different trajectory segment length (i.e., $y_{ij}$) values.}    \label{Fig:ORHCP3SolWithy_ij}
\end{figure}

Figure \ref{Fig:ORHCP3SolWithAlpha} confirms that by increasing the weight factor $\alpha$ (i.e., by giving more weight to the energy objective), we can constrain the agent tangential velocity and acceleration profiles. A converse behavior can be seen in Fig. \ref{Fig:ORHCP3SolWithR_j} with respect to the next-visit target $j$'s initial uncertainty $R_j(t_s)$. In particular, when $R_j(t_s)$ is high, the agent is required to arrive at target $j$ quickly (resulting in high tangential velocity and acceleration levels). In contrast, Fig. \ref{Fig:ORHCP3SolWithy_ij} reveals that when the trajectory segment length $y_{ij}$ is varied, the agent may not try to significantly regulate: (i) the arrival time at target $j$ (i.e., the transit-time $\rho_{ij}=t_f$) or (ii) the magnitude of the maximum tangential acceleration.

\subsection{Numerical Results for a \textbf{RHCP1}}

Similar to before, to numerically evaluate quantities relating to a \textbf{RHCP1} and its solution, we use the same parameter values mentioned before, along with an additional \emph{default} (initial) value $R_i(t_s)=50$. Figures \ref{Fig:ORHCP1SolWithAlpha} - \ref{Fig:ORHCP1SolWithR_i} respectively show how the \textbf{RHCP1} solution (i.e., $v_a(t), u_a(t), J_{sH}, J_{eH}$ and $J_H$) changes when the four parameters $\alpha, R_j(t_s), y_{ij}$ and $R_i(t_o)$ are varied.

\begin{figure}[!b]
     \centering
     \begin{subfigure}[b]{0.32\columnwidth}
         \centering
         \includegraphics[width = \textwidth]{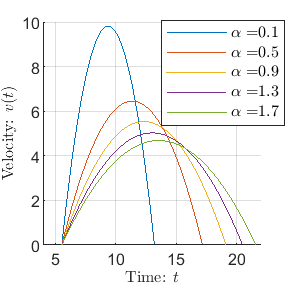}
         \caption{Tan. Vel.:$v_a(t)$}
         
     \end{subfigure} 
     \hfill
     \begin{subfigure}[b]{0.32\columnwidth}
         \centering
         \includegraphics[width = \textwidth]{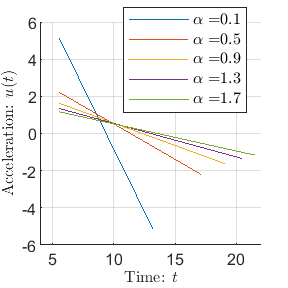}
         \caption{Tan. Acc.: $u_a(t)$}
         
     \end{subfigure} 
     \hfill
     \begin{subfigure}[b]{0.32\columnwidth}
         \centering
         \includegraphics[width = \textwidth]{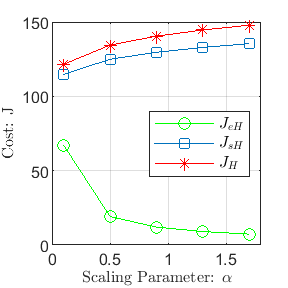}
         \caption{Cost: $J_{sH}, J_{eH}, J_{H}$}
         
     \end{subfigure}
     \caption{\textbf{RHCP1} solution under different $\alpha$ (in \eqref{Eq:MainObjective})  values.}
    \label{Fig:ORHCP1SolWithAlpha}
\end{figure}

\begin{figure}[!b]
     \centering
     \begin{subfigure}[b]{0.32\columnwidth}
         \centering
         \includegraphics[width = \textwidth]{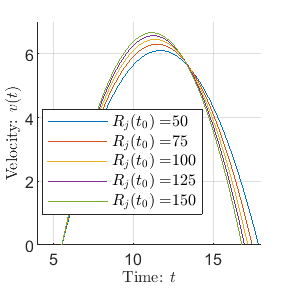}
         \caption{Tan. Vel.:$v_a(t)$}
         
     \end{subfigure} 
     \hfill
     \begin{subfigure}[b]{0.32\columnwidth}
         \centering
         \includegraphics[width = \textwidth]{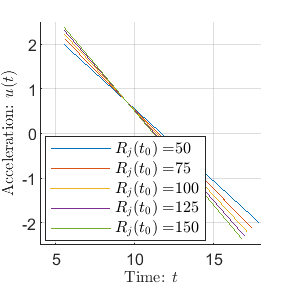}
         \caption{Tan. Acc.: $u_a(t)$}
         
     \end{subfigure} 
     \hfill
     \begin{subfigure}[b]{0.32\columnwidth}
         \centering
         \includegraphics[width = \textwidth]{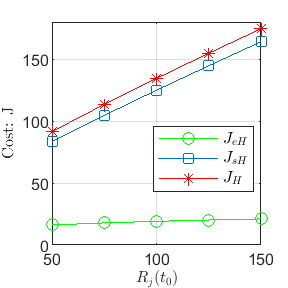}
         \caption{Cost: $J_{sH}, J_{eH}, J_{H}$}
         
     \end{subfigure}
     \caption{\textbf{RHCP1} solution under different $R_j(t_s)$ values.}
    \label{Fig:ORHCP1SolWithR_j}
\end{figure}

\begin{figure}[!t]
     \centering
     \begin{subfigure}[b]{0.32\columnwidth}
         \centering
         \includegraphics[width = \textwidth]{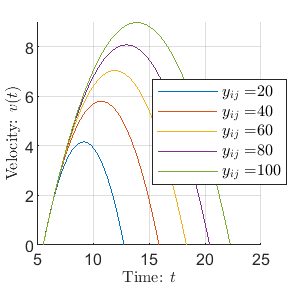}
         \caption{Tan. Vel.:$v_a(t)$}
         
     \end{subfigure} 
     \hfill
     \begin{subfigure}[b]{0.32\columnwidth}
         \centering
         \includegraphics[width = \textwidth]{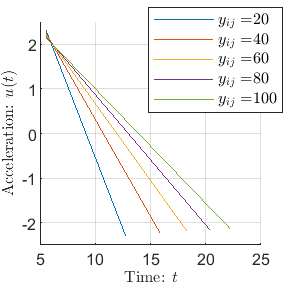}
         \caption{Tan. Acc.: $u_a(t)$}
         
     \end{subfigure} 
     \hfill
     \begin{subfigure}[b]{0.32\columnwidth}
         \centering
         \includegraphics[width = \textwidth]{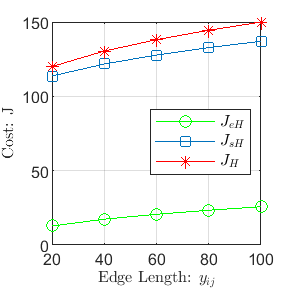}
         \caption{Cost: $J_{sH}, J_{eH}, J_{H}$}
         
     \end{subfigure}
     \caption{\textbf{RHCP1} solution under different $y_{ij}$ values.}    \label{Fig:ORHCP1SolWithy_ij}
\end{figure}

\begin{figure}[!t]
     \centering
     \begin{subfigure}[b]{0.32\columnwidth}
         \centering
         \includegraphics[width = \textwidth]{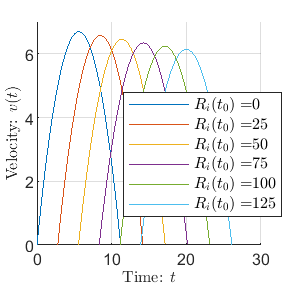}
         \caption{Tan. Vel.:$v_a(t)$}
         
     \end{subfigure} 
     \hfill
     \begin{subfigure}[b]{0.32\columnwidth}
         \centering
         \includegraphics[width = \textwidth]{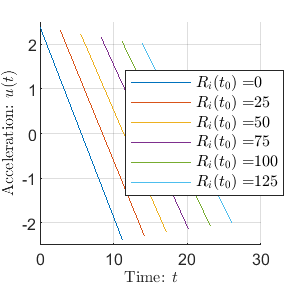}
         \caption{Tan. Acc.: $u_a(t)$}
         
     \end{subfigure} 
     \hfill
     \begin{subfigure}[b]{0.32\columnwidth}
         \centering
         \includegraphics[width = \textwidth]{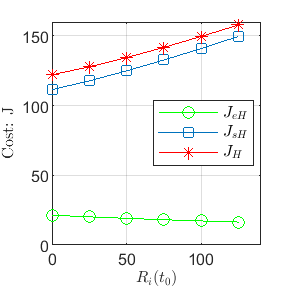}
         \caption{Cost: $J_{sH}, J_{eH}, J_{H}$}
         
     \end{subfigure}
     \caption{\textbf{RHCP1} solution under different $R_i(t_s)$ values.}
    \label{Fig:ORHCP1SolWithR_i}
\end{figure}

The \textbf{RHCP1} solution properties illustrated in Figs. \ref{Fig:ORHCP1SolWithAlpha}-\ref{Fig:ORHCP1SolWithy_ij} are identical to those of the \textbf{RHCP3} (shown in Figs. \ref{Fig:ORHCP3SolWithAlpha}-\ref{Fig:ORHCP3SolWithy_ij}), except for the fact that now $t_o>t_s=0$ (recall that $t_o$ is the planned time to leave the target $i$). However, Figs. \ref{Fig:ORHCP1SolWithAlpha}-\ref{Fig:ORHCP1SolWithy_ij} imply that $t_o$ is independent of the $\alpha, R_j(t_s)$ or $y_{ij}$ value. In contrast, Fig. \ref{Fig:ORHCP1SolWithR_i} reveals that $t_o$ is directly proportional to the $R_i(t_s)$ value. Moreover, Fig. \ref{Fig:ORHCP1SolWithR_i} shows that the maximum values of tangential velocity and acceleration decreases by a  small margin when $R_i(t_s)$ is increased. This implies that the agent plans to travel less urgently when it has to do more ``sensing'' at the current target $i$.

\subsection{Overall Performance in PMN Problems}

\begin{table*}[!t]
\centering
\caption{A comparison of performance metrics: $J_e$, $J_s$, $J_T$ (defined in \eqref{Eq:MainObjective}) and  $v_{max}$, $u_{max}$ (defined in \eqref{Eq:AgentPerformanceMetrics}) observed under different agent control methods SO, FO-1, FO-2 and FO-3 for each PMN Problem Configuration (PC) shown in Fig. \ref{Fig:FinalConditions}.}
\resizebox{\textwidth}{!}{
\begin{tabular}{|c|r|r|r|r|r|r|r|r|r|r|r|r|r|r|r|r|r|r|r|r|} \hline
\multirow{2}{*}{\textbf{PC}} & \multicolumn{4}{c|}{$J_T$}                                                                                       & \multicolumn{4}{c|}{$J_e/10000$}                                                                             & \multicolumn{4}{c|}{$J_s$}                                                                                   & \multicolumn{4}{c|}{$v_{max}$}                                                                                 & \multicolumn{4}{c|}{$u_{max}$}                                                                                  \\ \cline{2-21}
                  & \multicolumn{1}{c|}{SO} & \multicolumn{1}{c|}{FO-1} & \multicolumn{1}{c|}{FO-2} & \multicolumn{1}{c|}{FO-3} & \multicolumn{1}{c|}{SO} & \multicolumn{1}{c|}{FO-1} & \multicolumn{1}{c|}{FO-2} & \multicolumn{1}{c|}{FO-3} & \multicolumn{1}{c|}{SO} & \multicolumn{1}{c|}{FO-1} & \multicolumn{1}{c|}{FO-2} & \multicolumn{1}{c|}{FO-3} & \multicolumn{1}{c|}{SO} & \multicolumn{1}{c|}{FO-1} & \multicolumn{1}{c|}{FO-2} & \multicolumn{1}{c|}{FO-3} & \multicolumn{1}{c|}{SO} & \multicolumn{1}{c|}{FO-1} & \multicolumn{1}{c|}{FO-2} & \multicolumn{1}{c|}{FO-3}  \\ \hline
PC1               & \textbf{457.4}                   & 776.6                     & 593.2                     & 457.6                     & 165.7                   & 317.9                     & 229.7                     & \textbf{163.9}                    & 103.8                   & \textbf{98.5}                      & 103.2                     & 108.0                     & 84.6                    & 84.6                      & \textbf{70.8}                     & 82.1                      & 87.9                    & 87.9                      & 87.9                      & \textbf{62.2}                       \\ \hline
PC2               & 363.1                   & 1084.5                    & 252.9                     & \textbf{359.9}                    &    155.1                     & 156.1                     & 97.2                      & \textbf{153.1}                     & \textbf{32.2}                    & 751.5                     & 45.5                      & 33.4                      & 127.5                   & 127.5                     & \textbf{66.1}                     & 123.8                     & 87.4                    & 87.4                      & 87.4                      & \textbf{61.8}                       \\ \hline
PC3               & \textbf{1013.9}                  & 2090.8                    & 1489.3                    & 1020.1                    & \textbf{445.9}                   & 955.1                     & 671.5                     & 448.2                     & 62.7                    & \textbf{53.2}                      & 56.7                      & 63.9                      & 146.2                   & 120.5                     & \textbf{85.2}                      & 135.3                     & 124.5                   & 124.5                     & 124.5                     & \textbf{88.0}                       \\ \hline
PC4               & 705.7                   & 1212.8                    & 721.4                     & \textbf{704.9}                     & 306.0                   & 542.6                     & 310.7                     & \textbf{305.1}                     & \textbf{52.8}                    & 55.3                      & 58.6                      &  54.1                         & 126.4                   & 108.6                     & \textbf{76.8}                      & 122.9                     & 113.0                   & 113.0                     & 113.0                     & \textbf{79.9}                       \\ \hline
PC5               & \textbf{614.2}                   & 713.3                     & 655.4                     & 629.9                     & \textbf{66.9}                    & 120.4                     & 86.2                      & 69.4                      & 471.4                   & \textbf{456.5}                     &  471.5                         & 481.8                     & 157.4                   & 129.5                     & \textbf{91.7}                      & 140.0                     & 144.2                   & 144.2                     & 144.2                     & \textbf{102.0}                      \\ \hline
PC6               & \textbf{432.4}                   & 784.4                     & 706.1                     & 435.7                     & 141.3                   & 314.1                     & 268.6                     & \textbf{140.5}                     & 131.0                   & \textbf{114.4}                     & 133.2                     & 135.9                     & 94.8                    & 94.8                      & \textbf{86.7}                      & 91.1                      & 75.4                    & 75.4                      & 75.4                      & \textbf{53.3}                       \\ \hline
PC7               & \textbf{493.1}                   & 1093.6                    & 738.1                     & 497.4                     & 180.9                   & 468.1                     & 298.2                     & \textbf{180.7}                     & 107.2                   &  \textbf{94.9}                         & 102.0                     & 112.0                     & 106.6                   & 104.9                     & \textbf{74.2}                      & 98.1                      & 113.4                   & 113.4                     & 113.4                     & \textbf{80.2}                       \\ \hline
PC8               & \textbf{615.5}                   & 1503.3                    & 1069.6                    & 621.0                     & \textbf{248.3}                   & 669.8                     & 465.6                     & 249.5                     & 85.8                    & \textbf{74.4}                      & 76.3                      & 88.7                      & 103.5                   & 103.5                     & \textbf{81.7}                      & 100.5                     & 124.4                   & 124.4                     & 124.4                     & \textbf{87.9}                       \\ \hline
Avg.           & \textbf{586.9}                   & 1157.4                    & 778.2                     & 590.8                     & \textbf{213.8}                   & 443.0                     & 303.5                     & 213.8                     & \textbf{130.9}                   & 212.3                     & 130.9                     & 134.7                     & 118.4                   & 109.2                     & \textbf{79.2}                      & 111.7                     & 108.8                   & 108.8                     & 108.8                     &    \textbf{76.9}                        \\ \hline
\end{tabular}
}
\label{Tab:Results}
\end{table*}

In this final section, we compare the performance metrics $J_T, J_e$ and $J_s$ defined in \eqref{Eq:MainObjective} obtained for several different PMN problem configurations using agents behaving under: (i) SO, (ii) FO-1, (iii) FO-2 and (iv) FO-3 methods. In addition to $J_T, J_e$ and $J_s$, we also use the performance metrics:
\begin{equation}\label{Eq:AgentPerformanceMetrics}
    v_{max} \triangleq \max_{\substack{a\in\mathcal{A},\,t\in[0,T]}} v_a(t)\ \mbox{ and }\ 
    u_{max} \triangleq \max_{\substack{a\in\mathcal{A},\,t\in[0,T]}}
    \vert u_a(t)\vert,
\end{equation}
to represent overall agent behaviors rendered by different agent models. The proposed RHC solution to the PMN problem (under SO, FO-1, FO-2 and FO-3 methods) has been implemented in a JavaScript based simulator available at \href{http://www.bu.edu/codes/simulations/shiran27/PersistentMonitoring/}{http:www.bu.edu/codes/ simulations/shiran27/PersistentMonitoring/}.

\begin{figure}[!t]
     \centering
     \vspace{-3mm}
     \begin{subfigure}[b]{0.24\columnwidth}
         \centering
         \includegraphics[width=\textwidth]{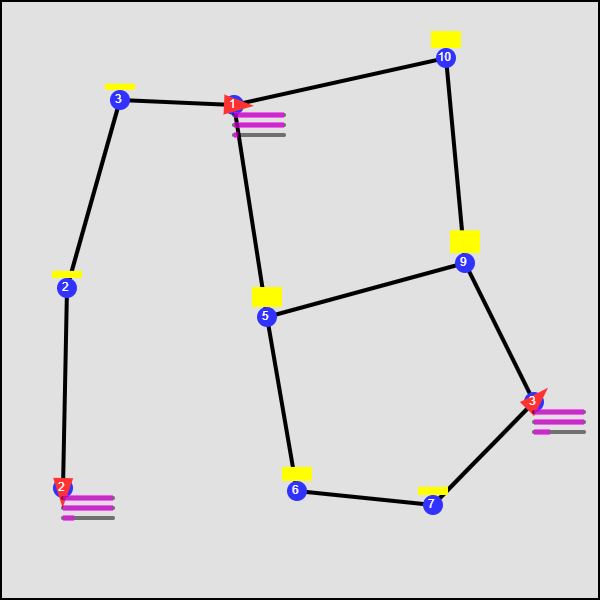}
         \caption{PC1}
         \label{SubFig:1}
     \end{subfigure}
     \begin{subfigure}[b]{0.24\columnwidth}
         \centering
         \includegraphics[width=\textwidth]{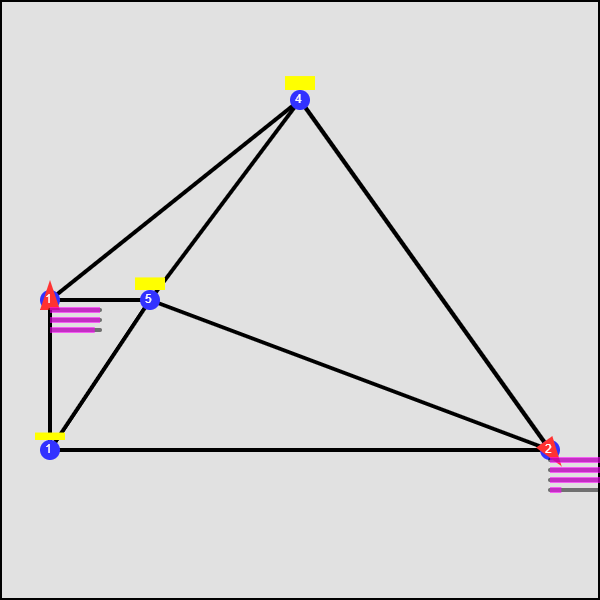}
         \caption{PC2}
         \label{SubFig:2}
     \end{subfigure}
     \begin{subfigure}[b]{0.24\columnwidth}
         \centering
         \includegraphics[width=\textwidth]{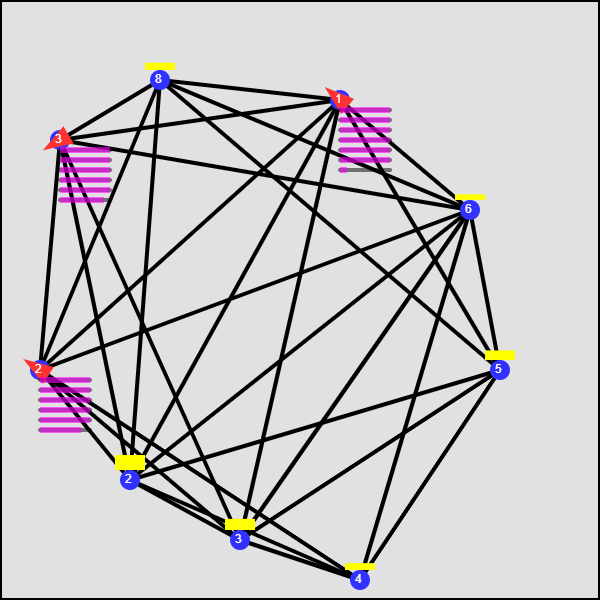}
         \caption{PC3}
         \label{SubFig:3}
     \end{subfigure}
     \begin{subfigure}[b]{0.24\columnwidth}
         \centering
         \includegraphics[width=\textwidth]{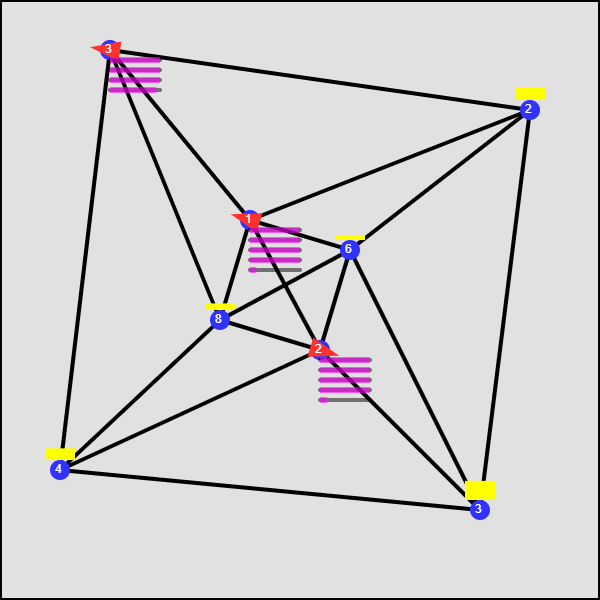}
         \caption{PC4}
         \label{SubFig:4}
     \end{subfigure}
     \begin{subfigure}[b]{0.24\columnwidth}
         \centering
         \includegraphics[width=\textwidth]{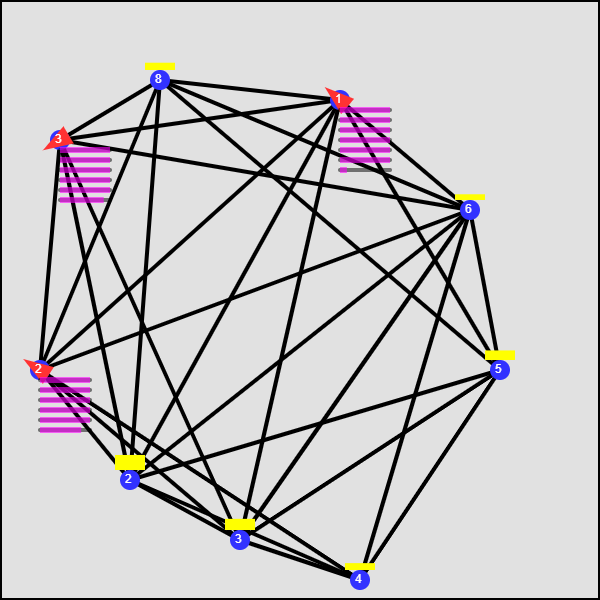}
         \caption{PC5}
         \label{SubFig:5}
     \end{subfigure}
     \begin{subfigure}[b]{0.24\columnwidth}
         \centering
         \includegraphics[width=\textwidth]{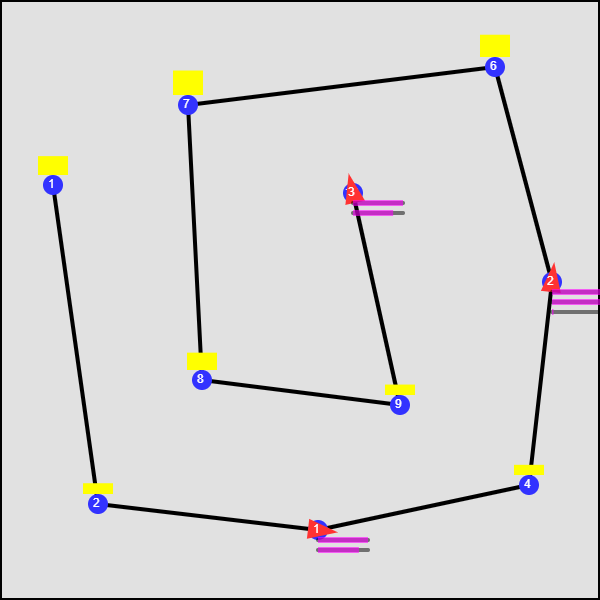}
         \caption{PC6}
         \label{SubFig:6}
     \end{subfigure}
     \begin{subfigure}[b]{0.24\columnwidth}
         \centering
         \includegraphics[width=\textwidth]{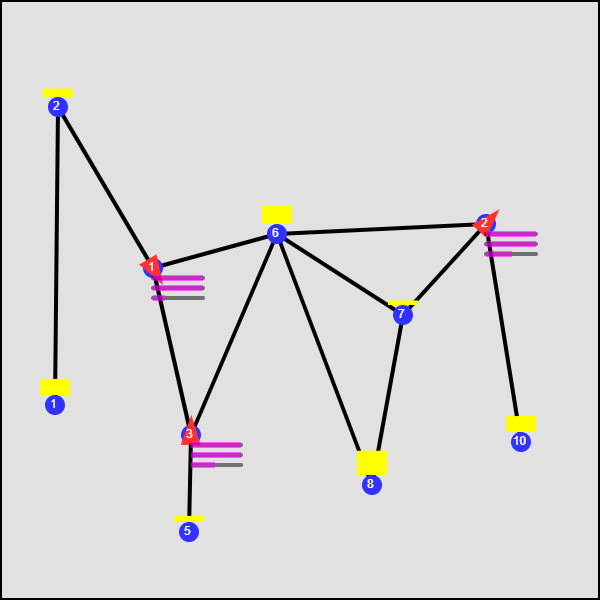}
         \caption{PC7}
         \label{SubFig:7}
     \end{subfigure}
     \begin{subfigure}[b]{0.24\columnwidth}
         \centering
         \includegraphics[width=\textwidth]{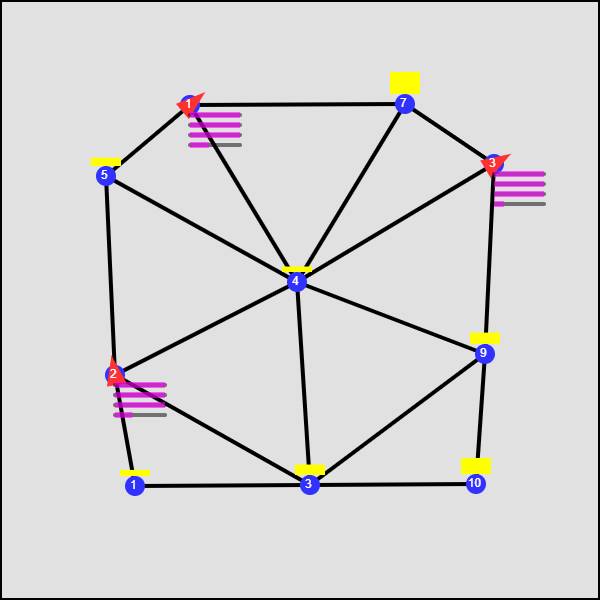}
         \caption{PC8}
         \label{SubFig:8}
     \end{subfigure}
     
        \caption{Final state of the PCs after using the highest performing agent model (and the respective optimal control): FO-3 in PC2 and PC4, and SO in all the other PCs.}
        \label{Fig:FinalConditions}
\end{figure}

In particular, we consider the eight multi-agent PMN problem configurations (PCs) shown in Fig. \ref{Fig:FinalConditions} (labeled PC1-PC8). In each PC, blue circles represent targets and black lines represent the trajectory segments available for the agents to travel between targets. Yellow vertical bars, purple horizontal bars and red triangles indicate target uncertainty levels, agent energy consumption levels and agent locations (i.e., $R_i(t),J_a(0,t)$ and $s_a(t)$), respectively. Since these three quantities are time-dependent, in the figures, only their terminal state (i.e., at $t=T$) is shown when the highest performing agent model (control method) is used.

The parameters of each PC have been chosen as follows: $A_{i}=1,\ B_{i}=10,\ R_{i}(0)=0.5,\ \forall i\in\mathcal{T}$ and target locations (i.e., $Y_i$) are specified in each PC figure. In all PCs, targets have been placed inside a $600\times600$ mission space. The time horizon was set to $T=500$. The initial locations of the agents were chosen such that $s_{a}(0)=Y_{i}$ with $i=1+(a-1)\ast\mathrm{round}(M/N)$. The upper bound on the planning horizon (i.e., $H$) was chosen as $H=\frac{T}{2}=250$ and the weight factor $\alpha$ in \eqref{Eq:MainObjective} was chosen as $\alpha =  213.3\times10^{-6}$.

Obtained comparative results are summarized in Tab. \ref{Tab:Results}. According to these results, on average, the energy-aware second-order agents (i.e., the SO method) have outperformed the energy-agnostic and energy-aware versions of first-order agents (i.e., FO-1 and FO-2, FO-3 methods, respectively) in terms of sensing objective $J_s$, energy objective $J_e$ as well as the total objective $J_T$.

However, the energy-aware (approximate) first-order agent control method FO-3 has shown relatively closer (within $1.17\%$ on average) performance levels to those of the SO method. Moreover, the FO-3 method has outperformed the SO method in terms of the performance metrics $u_{max}$ and $v_{max}$. This observation is reasonable because the motivation behind developing the FO-3 method was to improve the agent energy consumption. Recall also that we already have proven in Theorem \ref{Th:FO3MethodOptimalControls} that $v_{F3} = \frac{1}{k}v_{SO}$, $u_{F3} = \frac{3}{4k^2}v_{SO}$ with $k\geq1$ - which implies that agents under FO-3 method show lower maximum velocity and acceleration levels compared to agents under SO method. Nevertheless, we point out that SO and FO-3 methods have identical computational costs as their respective computational bottlenecks are in solving the non-linear equations \eqref{Eq:OptimumTransitTimeGeneral} and \eqref{Eq:FO3MethodOptimalControls1} - that has the same form.

\section{Conclusion}
\label{Sec:Conclusion} 

This paper considers the persistent monitoring problem defined on a network of targets that needs to be monitored by a team of energy-aware dynamic agents. Starting from an existing event-driven receding horizon control (RHC) solution, we exploit optimal control techniques to incorporate agent dynamics and agent energy consumption into the RHC problem setup. The proposed overall RHC solution is computationally efficient, distributed, on-line and gradient-free. Numerical results are provided to highlight the improvements with respect to an RHC solution that uses energy-agnostic first-order agents. Ongoing work aims to combine the proposed solution with a path planning algorithm to address situations where the agent trajectory segment shapes have to be optimally determined.


\appendix

\subsection{Selecting the Weight Factor: $\alpha$}
\label{App:NormalizationFactor}
The weight factor $\alpha$ present in both the main objective $J_T$ \eqref{Eq:MainObjective} and the RHCP objective $J_H$ \eqref{J_H} is an important factor that decides the trade-off between energy objective and the sensing objective components (i.e., $J_{eH}$ and $J_{sH}$, respectively, in the latter case). Moreover, note that $\alpha$ can be used to bound the resulting optimal agent velocities and accelerations from the proposed RHC solution. Therefore, it is important to have an intuitive technique to select (and vary) $\alpha\in [0,\infty)$. 

To develop such a technique, we use the RHCP form \eqref{J_H}: $J_H = \alpha J_{eH} + J_{sH}$, rather than the main optimization problem form \eqref{Eq:MainObjective}. A typical RHCP objective function $J_H$ that considers both energy and sensing objectives (i.e., $J_{eH}$ and $J_{sH}$, respectively) can be written as 
\begin{equation}
    J_H = \beta \frac{J_{eH}}{E_{H}^{max}} + (1-\beta) \frac{J_{sH}}{S_{H}^{max}},
\end{equation}
where $E_H^{max}$ and $S_H^{max}$ are upper-bounds to the terms $J_{eH}$ and $J_{sH}$ respectively and $\beta$ is a parameter such that $\beta \in [0,1]$. Next, let us re-arrange the above expression to isolate the sensing objective component as
\begin{equation}\label{Eq:NormalizedJ_H}
    J_H = \underbrace{\left[\frac{\beta}{1-\beta}\frac{S_H^{max}}{E_H^{max}}\right]}_{\alpha} J_{eH} + J_{sH}.
\end{equation}
Now, if the ratio $S_H^{max}/E_H^{max}$ is known, a candidate $\alpha$ value can be obtained intuitively by selecting $\beta\in[0,1)$ appropriately. For an example, selecting $\beta=0.5$ gives an equal weight to both energy and sensing objective components.  

To estimate the ratio between $S_H^{max}$ and $E_H^{max}$ we can consider a simple RHCP that occurs when an agent $a$ is ready to leave a target $i$ with a (single) neighboring target $j$ connected through a trajectory segment $(i,j)$. For such a scenario, assuming steady-state operation, using Theorem \ref{Th:ContributionTarget}, we can show that $S_{H}^{max} \propto \rho_{ij}$. Next, let us define quantities $E_H^{max}$, $v_{max}$ and $u_{max}$ based on Theorem \ref{Th:SOMethodOptimalControls} \eqref{Eq:SOMethodMaxValues} as 
$$ E_H^{max} \triangleq E_{SO} \propto \frac{y_{ij}^2}{\rho_{ij}^3},\ \ 
v_{max} \triangleq v_{SO}\propto \frac{y_{ij}}{\rho_{ij}},\ \ 
u_{max} \triangleq u_{SO}\propto \frac{y_{ij}}{\rho_{ij}^2}.
$$
Combining $S_H^{max}$, $E_H^{max}$ together with $v_{max}$ or (alternatively) $u_{max}$ stated above, we can show that 
\begin{eqnarray}
\frac{S_H^{max}}{E_H^{max}} \propto \frac{y_{ij}^2}{v_{max}^4}\ \ \ 
\mbox{ or }\ \ \ 
\frac{S_H^{max}}{E_H^{max}} \propto \frac{1}{u_{max}^2},
\end{eqnarray}
respectively. Here, $v_{max}$ and $u_{max}$ can be thought of as the preferred tangential velocity and acceleration bounds for the agents, respectively. And $y_{ij}$ can be thought of as the mean trajectory segment length over all $(i,j)\in\mathcal{E}$. Finally, neglecting the constants of proportionality in the above statements and using \eqref{Eq:NormalizedJ_H}, we can state $\alpha$ as 
\begin{equation}\label{Eq:CandidateAlpha}
    \alpha = \frac{\beta}{1-\beta} \frac{y_{ij}^2}{v_{max}^4} 
    \ \ \ \mbox{ or } \ \ \ 
    \alpha = \frac{\beta}{1-\beta} \frac{1}{u_{max}^2}.
\end{equation}

This result \eqref{Eq:CandidateAlpha} provides a systematic way to select $\alpha$ while accounting for: (i) the relative balance between sensing and energy objectives (via $\beta\in[0,1]$) and (ii) the preferred tangential velocity and acceleration bounds (via $v_{max}$ and  $u_{max}$, respectively). For an example, with $\beta = 0.5,\ v_{max}=50,\ y_{ij}=25$, the relationship in \eqref{Eq:CandidateAlpha} gives $\alpha = 1.01\times 10^{-4}$.

\subsection{Proof of Theorem \ref{Th:AngularVelocity}}
\label{App:AngularVelocityProof}

First, we transform the parametric form $\{(x(p),y(p)):p\in[p_o,p_f]\}$ of the trajectory segment shape in to the form $\{(x(l),y(l)):l\in[0,y_{ij}]\}$ where $l$ represents the distance along the trajectory segment from  $(x(p_o),y(p_o))$ to $(x(p),y(p))$,\ $p\in[p_o,p_f]$ (recall that $y_{ij}$ is the total length of the interested trajectory segment $(i,j)\in\mathcal{E}$). To achieve the said transformation, we should be able to express the parameter $p$ explicitly in terms of the distance $l$. For this purpose, exploiting the geometry (see also Fig. \ref{Fig:AgentTrajectory2}), we can write a differential equation: 
\begin{equation}\label{Eq:TrajectoryParametrizationTransform}
dl = \sqrt{(x_p^\prime)^2 + (y_p^\prime)^2}\, dp.
\end{equation}
Under assumption \ref{As:AngularVelocity}, \eqref{Eq:TrajectoryParametrizationTransform} can be solved to obtain explicit relationships: $l=f(p)$ and $p=f^{-1}(l)$, where $f:[p_o,p_f]\rightarrow [0,y_{ij}]$ is as in \eqref{Eq:ParametrizationFunction1}. Thus, we now can express the trajectory segment shape in the form: $\{(x(l),y(l)):l\in[0,y_{ij}]\}$.


Second, according to Fig. \ref{Fig:AgentTrajectory2}, note that when the agent $a\in\mathcal{A}$ is at $s_a(t) \equiv (x(l),y(l))$, its orientation $\theta$ satisfies
\begin{equation}
    \tan \theta = \frac{\dot{y}(l)}{\dot{x}(l)} = \frac{\frac{dy(l)}{dl}\frac{dl}{dt}}{\frac{dx(l)}{dl}\frac{dl}{dt}} = \frac{y'}{x'}.
\end{equation}
In the above, the notation ``\,$'$\,'' (without a subscript) has been used here to represent the $\frac{d\cdot}{dl}$ operator. The time derivative of this relationship gives 
\begin{equation}
    \sec^2\theta \, \frac{d\theta}{dt} = \frac{(x'y''-y'x'')}{(x')^2}\,\frac{dl}{dt}. 
\end{equation}
Note that if $l=l_a(t)$ is used to represent the total distance the agent traveled on the trajectory segment by the time $t\in[t_o,t_f]$, we can also write $\frac{dl}{dt} = v_a(t)$ (i.e., the agent tangential velocity) and $\frac{d\theta}{dt} = w_a(t)$ (i.e., the agent angular velocity). Therefore, using the above two relationships and the trigonometric identity: $\sec^2\theta = 1 + \tan^2\theta$, we can obtain $w_a(t)$ for any $t\in[t_o,t_f]$ as
\begin{equation}\label{Eq:AngularVelocity2}
    w_a(t) = \underbrace{\frac{x'y''-y'x''}{(x')^2 + (y')^2}}_{G(l)}\,v_a(t). 
\end{equation}
Here, note that the first term $G(l)$ is a function of $l=l_a(t)$.  

Finally, we can transform this $G(l)$ term in \eqref{Eq:AngularVelocity2} to obtain a function of parameter $p$, using the following relationships (from the chain rule and the fact that $l=f(p)$):
\begin{equation}\label{Eq:ChainRules1}
    x^\prime = \frac{x_p^\prime}{f_p^\prime},\ 
    y^\prime = \frac{y_p^\prime}{f_p^\prime},\ 
    x^{\prime\prime} = \frac{x_p^{\prime\prime}f_p^\prime -x_p^\prime f_p^{\prime\prime}}{(f_p^\prime)^3},\ 
    y^{\prime\prime} = \frac{y_p^{\prime\prime}f_p^\prime -y_p^\prime f_p^{\prime\prime}}{(f_p^\prime)^3},
\end{equation}
and (using \eqref{Eq:ParametrizationFunction1})
\begin{equation}\label{Eq:ChainRules2}
    f_p^\prime = \sqrt{(x_p^\prime)^2 + (y_p^\prime)^2}.
\end{equation}
Recall that, in the above, we have used the notation ``\,$^\prime$\,'' (with a subscript $p$) to denote the operator $\frac{d\cdot}{dp}$. Now, using \eqref{Eq:ChainRules1} and \eqref{Eq:ChainRules2}, $G(l)$ in \eqref{Eq:AngularVelocity2} can be written as  
\begin{equation}
    G(l) = G(f(p)) = \frac{x_p^\prime y_p^{\prime\prime} - y_p^{\prime} x_p^{\prime\prime}}{\left((x_p^\prime)^2 + (y_p^\prime)^2\right)^{\frac{3}{2}}}. 
\end{equation}
Comparing the above result with \eqref{Eq:ParametrizationFunction2}, notice that $G(l)=F(p)$. Therefore, \eqref{Eq:AngularVelocity2} can be written as $w_a(t) = F(p)v_a(t)$ where now $p$ can be replaced with $p=f^{-1}(l) = f^{-1}(l_a(t))$ to obtain \eqref{Eq:AngularVelocity}: 
$$
    w_a(t) = F(f^{-1}(l_a(t)))\,v_a(t), 
$$
which completes the proof.

\subsection{Proof of Proposition \ref{Pr:FO1}}
\label{App:Lm_FO1_Proof}
It is easy to show that $v_{F1}^{ij}(u_{SO}^{max},v_m)$ \eqref{Eq:FO1MethodMaxEdgeVelocity} is a monotonically increasing function with respect to $v_m$. In particular, if $v_m>\frac{1}{2}\sqrt{y_{ij}u_{SO}^{max}}$, the function $v_{F1}^{ij}(u_{SO}^{max},v_m)$ plateaus at a level $\sqrt{y_{ij}u_{SO}^{max}}$.   Therefore, the set of $v_m$ values that satisfies the inequality $v_{F1}^{ij}(u_{SO}^{max},v_m) \leq v_{SO}^{max}$ can be stated as 
\begin{equation}
    v_m \leq v_m^{ij} \triangleq 
    \begin{cases}
    \frac{y_{ij}u_{SO}^{max}v_{SO}^{max}}{(v_{SO}^{max})^2 + y_{ij}u_{SO}^{max}}\ \mbox{ if } y_{ij} \geq \frac{(v_{SO}^{max})^2}{u_{SO}^{max}}\\
    \infty \ \mbox{ otherwise. }
    \end{cases}
\end{equation}
According to \eqref{Eq:FO1MethodParameters}, the inequality $v_{F1}^{ij}(u_{SO}^{max},v_m) \leq v_{SO}^{max}$ should hold for all $(i,j)\in\mathcal{E}$. Therefore, the feasible set of $v_m$ in \eqref{Eq:FO1MethodParameters} is: $v_m \leq \min_{(i,j)\in\mathcal{E}}v_m^{ij}$. Again, using the monotonicity property of $v_{F1}^{ij}(u_{SO}^{max},v_m)$ (which is also the the objective function of \eqref{Eq:FO1MethodParameters}), we can show that the optimal $v_m$ value (i.e., $v_{m1}$) of \eqref{Eq:FO1MethodParameters} is the maximum feasible $v_m$ value, i.e., 
\begin{equation}
\begin{alignedat}{4}
& v_{m1} = \underset{\makebox[1.2cm]{\footnotesize $(i,j)\in\mathcal{E}$}}{\min}v_m^{ij}\  = & \underset{\makebox[2cm]{\footnotesize $(i,j)\in\mathcal{E}$}}{\min} & \ \ && \frac{y_{ij}u_{SO}^{max}v_{SO}^{max}}{(v_{SO}^{max})^2+y_{ij}u_{SO}^{max}} & \\
&  & \makebox[2cm]{subject to} &  && y_{ij}\geq  (v_{SO}^{max})^2/u_{SO}^{max}. 
\end{alignedat}
\end{equation}

\subsection{Second-Order Agent Models with Constraints}
\label{SubSec:SOAgentsWithConstraints}
In this section, we show how a second-order agent $a\in\mathcal{A}$ should select its behavior (including the transit-time) on a trajectory segment when solving a RHCP under tangential velocity or acceleration bounds.



\paragraph{\textbf{SO-V Method}}
The SO-V method assumes that the agent tangential velocity is bounded such that $\vert v_a(t) \vert \leq \bar{v}$ where $\bar{v}$ is predefined and satisfies $\bar{v}<v_{SO} =\frac{3y_{ij}}{2\rho_{SO}}$ (recall that $\rho_{SO}$ is the optimal transit-time found for the unconstrained SO method). Based on the optimal unconstrained velocity profile \eqref{Eq:SOMethodProfiles}, we can expect the optimal constrained velocity profile to contain three different phases: two quadratic segments at the beginning and the end and a constant velocity segment in the middle, as shown in Fig. \ref{Fig:SO-V}. 

A generalized version of the optimal unconstrained velocity profile \eqref{Eq:SOMethodProfiles} can be written as $v(t) = \alpha_0 t (\beta - t),\ t\in[0,\beta]$ where $\beta$ can be thought of as controllable parameter and $\alpha_0 = \frac{6y_{ij}}{\beta^3}$ (enforcing the condition: $\int_0^\beta v(t)dt = y_{ij}$). We next use this $v(t)$ profile to construct the optimal constrained velocity profile $v_a(t)$ as
\begin{equation}\label{Eq:SO2VelocityProfile}
    v_a(t) \triangleq 
    \begin{cases}
    v(t) &t\in [0,t_1)\\
    \bar{v} &t\in[t_1,t_2)\\
    v(t-(\rho_{SV}-\beta)) &t\in[t_2,\rho_{SV}],
    \end{cases}
\end{equation}
where $t_1$ is such that $v(t_1)=\bar{v}$ (existence of such a $t_1$ is guaranteed when $\beta \leq \rho_{SO}$), $t_2 = \rho_{SV}-t_1$ (from symmetry) and the transit-time $\rho_{SV}$ is such that $\int_0^{\rho_{SV}} v_a(t)dt = y_{ij}$. In particular, it can be shown that 
\begin{eqnarray}
        t_1 &=& \frac{\beta}{2}\Big(1-\big(1-\frac{\beta}{t_v}\big)^{\frac{1}{2}}\Big),\nonumber\\\label{Eq:SO2rhoandbeta}
        \rho_{SV} &=& \beta + \frac{2t_v}{3}\big(1-\frac{\beta}{t_v}\big)^{\frac{3}{2}}, 
\end{eqnarray}
where $t_v \triangleq \frac{3y_{ij}}{2\bar{v}}$. We highlight that the agent velocity profile $v_a(t)$ defined in \eqref{Eq:SO2VelocityProfile} depends only on the parameter $\beta$.

\begin{figure}[!h]
    \centering
    \includegraphics[width=2in]{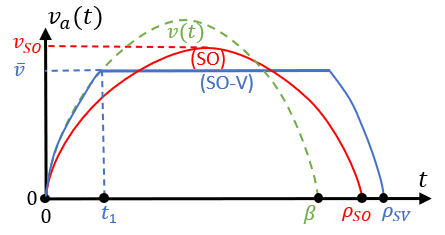}
    \caption{Tangential velocity profiles on a trajectory segment $(i,j)\in\mathcal{E}$ under unconstrained (SO) and constrained (SO-V) second-order agent models.}
    \label{Fig:SO-V}
\end{figure}

Under the SO-V method, the sensing objective component of the OCP \eqref{Eq:OCP} is $J_{sH}^*(\rho_{SV})$ and the energy objective component of the OCP can be written as 
\begin{align}
    E_{SV} \triangleq& \int_0^{\rho_{SV}}(\frac{dv_a(t)}{dt})^2 dt 
    =   \frac{12y_{ij}^2}{\beta^3}
    \Big( \big(1-\frac{\beta}{t_v}\big)^{\frac{3}{2}} - 1 \Big).
\end{align}
Therefore, the OCP objective that needs to be optimized in a RHCP under the SO-V method is 
\begin{equation}\label{Eq:SO2CompositeObjective}
    J_H = \alpha E_{SV} + J_{sH}^*(\rho_{SV}).
\end{equation}   
Thus, the optimal transit-time $\rho_{SV}$ (and hence the optimal $\beta$ value via \eqref{Eq:SO2rhoandbeta}) can be found using:
\begin{equation}
\frac{dJ_H}{d\rho_{SV}} = \alpha \frac{dE_{SV}}{d\beta}/\frac{d\rho_{SV}}{d\beta} + \frac{dJ_{sH}^*(\rho_{SV})}{d\rho_{SV}} = 0.   
\end{equation}
As shown in Fig. \ref{Fig:OverviewOfTheRHCPSolution}, note that finding the optimal transit-time corresponding to the OCP \eqref{Eq:OCP} enables determining the remaining control inputs in $U_{iaj}^*$ of the RHCP \eqref{Eq:RHCGenSolStep1}.

\paragraph{\textbf{SO-A Method}}
The SO-A method assumes that the agent tangential acceleration is bounded such that $\vert u_a(t) \vert \leq \bar{u}$ where $\bar{u}$ is predefined and satisfies $\bar{u}<u_{SO} = \frac{6y_{ij}}{\rho_{SO}^2}$. Based on the optimal unconstrained acceleration profile \eqref{Eq:SOMethodProfiles}, we can expect the optimal constrained acceleration profile to be a composition of three stages: two constant acceleration sessions at the beginning and the end and a linearly decreasing acceleration session in the middle, as shown in Fig. \ref{Fig:SO-A}. 

In particular, the optimal constrained acceleration profile $u_a(t)$ can be written as 
\begin{equation}\label{Eq:SO3Acceleration}
    u_a(t) \triangleq 
    \begin{cases}
    \bar{u} & t \in [0,t_1]\\
    \bar{u}-2\beta(t-t_1) & t \in [t_1,t_2]\\
    -\bar{u} & t \in [t_2,\rho_{SA}],
    \end{cases}
\end{equation}
where $t_1,t_2$ are switching times such that $v_a(t) = v_{SA}$ and $\rho_{SA}$ is the transit-time. Using the symmetry and the relationship $\int_0^{\rho_{SA}}v_a(t)dt = y_{ij}$, it can be shown that 
\begin{align}
    v_{SA} =& \sqrt{y_{ij}\bar{u}-\frac{\bar{u}^4}{6\beta^2}},\nonumber\\ \label{Eq:SO3rhoandbeta}
    \rho_{SA} =&
    \frac{2}{\bar{u}} \sqrt{y_{ij}\bar{u}-\frac{\bar{u}^4}{6\beta^2}} + \frac{\bar{u}}{\beta}.
\end{align}
Notice that $\beta$ is a controllable parameter that fully defines the optimal constrained acceleration profile in \eqref{Eq:SO3Acceleration}. 

\begin{figure}[!h]
    \centering
    \includegraphics[width=2in]{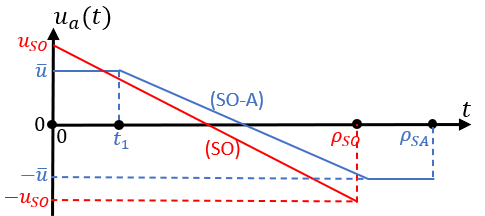}
    \caption{Tangential acceleration profiles on a trajectory segment $(i,j)\in\mathcal{E}$ under unconstrained (SO) and constrained (SO-A) second-order agent models.}
    \label{Fig:SO-A}
\end{figure}

Under the SO-A method, the sensing objective component of a RHCP is $J_{sH}^*(\rho_{SA})$ and the energy objective component of the OCP \eqref{Eq:OCP} can be written as 
\begin{align}
    E_{SA} \triangleq& \int_0^{\rho_{SA}}u_a^2(t) dt =   2\bar{u}\sqrt{y_{ij}\bar{u}-\frac{\bar{u}^4}{6\beta^2}} + \frac{\bar{u}^3}{3\alpha}.
\end{align}
Therefore, the composite objective function of the OCP under the SO-A method is 
\begin{equation}\label{Eq:SO3CompositeObjective}
    J_H = \alpha E_{SA} + J_{sH}^*(\rho_{SA}).
\end{equation}   
Thus, the optimal transit-time $\rho_{SA}$ (and hence the optimal $\beta$ value via \eqref{Eq:SO3rhoandbeta}) can be found using the equation:
\begin{equation}
\frac{dJ_H}{d\rho_{SA}} = \alpha \frac{dE_{SA}}{d\beta}/\frac{d\rho_{SA}}{d\beta} + \frac{dJ_{sH}^*(\rho_{SA})}{d\rho_{SA}} = 0.   
\end{equation}

\subsection{First-Order Agent Models with Constraints}
\label{SubSec:FOAgentsWithConstraints}
In this section, we investigate how a first-order agent $a\in \mathcal{A}$ should select its behavior (including the transit-time) on a trajectory segment when solving an RHCP under tangential velocity or acceleration bounds. 

\paragraph{\textbf{FO-V Method}}
The FO-V method assumes that the agent tangential velocity is bounded such that $v_a(t)\leq \bar{v}$ where $\bar{v}$ is predefined and satisfies $\bar{v}<v_{F3} = \frac{3y_{ij}}{2\rho_{F3}}$ (recall that $\rho_{F3}$ is the transit-time found for the unconstrained FO-3 method). 

Under this constrained setting, the optimal agent tangential velocity profile is shown in Fig. \ref{Fig:FO-V} where $u_{FV}$ is a controllable parameter. Taking the corresponding transit-time as $\rho_{FV}$ and using the fact that $\int_0^{\rho_{FV}}v_a(t)dt = y_{ij}$, it can be shown that 
\begin{equation}\label{Eq:FVAcceleration}
    u_{FV} = \frac{\bar{v}^2}{\bar{v}\rho_{FV} - y_{ij}}.
\end{equation}

Similar to before, under this FO-V method, the sensing objective component of the OCP \eqref{Eq:OCP} is $J_{sH}^*(\rho_{FV})$ and the energy objective component of a RHCP can be written as 
\begin{equation}
    E_{FV} = \frac{2\bar{v}^3}{\bar{v}\rho_{FV}-y_{ij}}.
\end{equation}
The composite objective function that needs to be optimized the OCP \eqref{Eq:OCP} under the FO-V method is 
\begin{equation}\label{Eq:FOVCompositeObjective}
    J_H = \alpha E_{FV} + J_{sH}^*(\rho_{FV}).
\end{equation}   
Therefore, the optimal transit-time $\rho_{FV}$ (and hence the optimal $u_{FV}$ value via \eqref{Eq:FVAcceleration}) can be found using:
\begin{equation}
\frac{dJ_H}{d\rho_{FV}} = \alpha \frac{dE_{FV}}{d\rho_{FV}} + \frac{dJ_{sH}^*(\rho_{FV})}{d\rho_{FV}} = 0.   
\end{equation}

\begin{figure}[!h]
    \centering
    \includegraphics[width=2in]{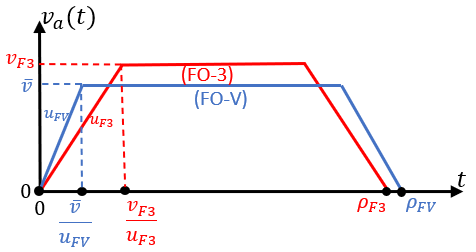}
    \caption{Tangential velocity profiles on a trajectory segment $(i,j)\in\mathcal{E}$ under unconstrained (FO-3) and constrained (FO-V) first-order agent models.}
    \label{Fig:FO-V}
\end{figure}

\paragraph{\textbf{FO-A Method}}
The FO-A method assumes that the agent tangential acceleration is bounded such that $\vert u_a(t)\vert \leq \bar{u}$ where $\bar{u}$ is predefined and satisfies $\bar{u} < u_{F3} = \frac{9y_{ij}}{2\rho_{F3}^2}$. 

Under this constrained setting, the optimal agent tangential velocity profile is shown in Fig. \ref{Fig:FO-A} where $v_{FA}$ is a controllable parameter. Taking the corresponding transit-time as $\rho_{FA}$ and using the fact that $\int_0^{\rho_{FA}}v_a(t)dt = y_{ij}$, it can be shown that 
\begin{equation}\label{Eq:FAVelocity}
    v_{FA} = \frac{\rho_{FA}\bar{u}}{2}-\frac{\bar{u}}{2}\sqrt{\rho_{FA}^2-\frac{4y_{ij}}{\bar{u}}}.
\end{equation}

Following the same procedure as before, under this FO-A method, the sensing objective component of the OCP is $J_{sH}^*(\rho_{FA})$ and the energy objective component of a RHCP can be written as 
\begin{equation}
    E_{FA} = \bar{u}^2\big(\rho_{FA}-\sqrt{\rho_{FA}^2-\frac{4y_{ij}}{\bar{u}}}\big).
\end{equation}
The composite objective function that needs to be optimized in a RHCP under the FO-A method is 
\begin{equation}\label{Eq:FOVCompositeObjective}
    J_H = \alpha E_{FA} + J_{sH}^*(\rho_{FA}).
\end{equation}   
Therefore, the optimal transit-time $\rho_{FA}$ (and hence the optimal $v_{FA}$ value via \eqref{Eq:FAVelocity}) can be found using:
\begin{equation}
\frac{dJ_H}{d\rho_{FA}} = \alpha \frac{dE_{FA}}{d\rho_{FA}} + \frac{dJ_{sH}^*(\rho_{FA})}{d\rho_{FA}} = 0.   
\end{equation}

\begin{figure}[!h]
    \centering
    \includegraphics[width=2in]{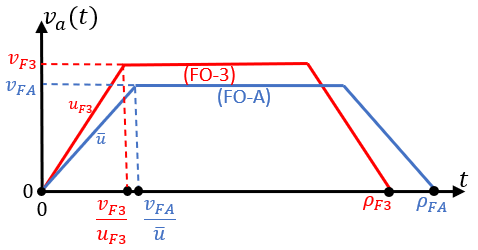}
    \caption{Tangential velocity profiles on a trajectory segment $(i,j)\in\mathcal{E}$ under unconstrained (FO-3) and constrained (FO-A) first-order agent models.}
    \label{Fig:FO-A}
\end{figure}









\bibliographystyle{IEEEtran}
\bibliography{References}

\begin{thebibliography}{10}
\providecommand{\url}[1]{#1}
\csname url@samestyle\endcsname
\providecommand{\newblock}{\relax}
\providecommand{\bibinfo}[2]{#2}
\providecommand{\BIBentrySTDinterwordspacing}{\spaceskip=0pt\relax}
\providecommand{\BIBentryALTinterwordstretchfactor}{4}
\providecommand{\BIBentryALTinterwordspacing}{\spaceskip=\fontdimen2\font plus
\BIBentryALTinterwordstretchfactor\fontdimen3\font minus
  \fontdimen4\font\relax}
\providecommand{\BIBforeignlanguage}[2]{{%
\expandafter\ifx\csname l@#1\endcsname\relax
\typeout{** WARNING: IEEEtran.bst: No hyphenation pattern has been}%
\typeout{** loaded for the language `#1'. Using the pattern for}%
\typeout{** the default language instead.}%
\else
\language=\csname l@#1\endcsname
\fi
#2}}
\providecommand{\BIBdecl}{\relax}
\BIBdecl

\bibitem{Elwin2020}
M.~L. Elwin, R.~A. Freeman, and K.~M. Lynch, ``{Distributed Environmental
  Monitoring with Finite Element Robots},'' \emph{IEEE Trans. on Robotics},
  vol.~36, no.~2, pp. 380--398, 2020.

\bibitem{Kingston2008}
D.~Kingston, R.~W. Beard, and R.~S. Holt, ``{Decentralized Perimeter
  Surveillance Using a Team of UAVs},'' \emph{IEEE Trans. on Robotics},
  vol.~24, no.~6, pp. 1394--1404, 2008.

\bibitem{Reshma2016}
R.~Reshma, T.~Ramesh, and P.~Sathishkumar, ``{Security Situational Aware
  Intelligent Road Traffic Monitoring Using UAVs},'' in \emph{Proc. of 2nd IEEE
  Intl. Conf. on VLSI Systems, Architectures, Technology and Applications},
  2016, pp. 1--6.

\bibitem{Smith2011}
S.~L. Smith, M.~Schwager, and D.~Rus, ``{Persistent Monitoring of Changing
  Environments Using a Robot with Limited Range Sensing},'' in \emph{Proc. of
  IEEE Intl. Conf. on Robotics and Automation}, 2011, pp. 5448--5455.

\bibitem{Yu2015}
J.~Yu, S.~Karaman, and D.~Rus, ``{Persistent Monitoring of Events With
  Stochastic Arrivals at Multiple Stations},'' \emph{IEEE Trans. on Robotics},
  vol.~31, no.~3, pp. 521--535, 2015.

\bibitem{Mathew2015}
N.~Mathew, S.~L. Smith, and S.~L. Waslander, ``{Multirobot Rendezvous Planning
  for Recharging in Persistent Tasks},'' \emph{IEEE Trans. on Robotics},
  vol.~31, no.~1, pp. 128--142, 2015.

\bibitem{Hari2019}
S.~K. Hari, S.~Rathinam, S.~Darbha, K.~Kalyanam, S.~G. Manyam, and D.~Casbeer,
  ``{The Generalized Persistent Monitoring Problem},'' in \emph{Proc. of
  American Control Conf.}, vol. 2019-July, 2019, pp. 2783--2788.

\bibitem{Welikala2020J4}
S.~Welikala and C.~G. Cassandras, ``{Event-Driven Receding Horizon Control For
  Distributed Persistent Monitoring in Network Systems},'' \emph{Automatica},
  vol. 127, p. 109519, 2021.

\bibitem{Wang2017}
Y.-W. Wang, Y.-W. Wei, X.-K. Liu, N.~Zhou, and C.~G. Cassandras, ``{Optimal
  Persistent Monitoring Using Second-Order Agents with Physical Constraints},''
  \emph{IEEE Trans. on Automatic Control}, vol.~64, no.~8, pp. 3239--3252,
  2017.

\bibitem{Zhou2019}
N.~Zhou, C.~G. Cassandras, X.~Yu, and S.~B. Andersson, ``{Optimal
  Threshold-Based Distributed Control Policies for Persistent Monitoring on
  Graphs},'' in \emph{Proc. of American Control Conf.}, 2019, pp. 2030--2035.

\bibitem{Lan2013}
X.~Lan and M.~Schwager, ``{Planning Periodic Persistent Monitoring Trajectories
  for Sensing Robots in Gaussian Random Fields},'' in \emph{In Proc. of IEEE
  Intl. Conf. on Robotics and Automation}, 2013, pp. 2415--2420.

\bibitem{Lin2013}
X.~Lin and C.~G. Cassandras, ``{An Optimal Control Approach to The Multi-Agent
  Persistent Monitoring Problem in Two-Dimensional Spaces},'' \emph{IEEE Trans.
  on Automatic Control}, vol.~60, no.~6, pp. 1659--1664, 2015.

\bibitem{Kirk2020}
\BIBentryALTinterwordspacing
J.~Kirk, ``{Traveling Salesman Problem - Genetic Algorithm},'' 2020. [Online].
  Available:
  \url{https://www.mathworks.com/matlabcentral/fileexchange/13680-traveling-salesman-problem-genetic-algorithm}
\BIBentrySTDinterwordspacing

\bibitem{Rezazadeh2019}
N.~Rezazadeh and S.~S. Kia, ``{A Sub-Modular Receding Horizon Approach to
  Persistent Monitoring for A Group of Mobile Agents Over an Urban Area},'' in
  \emph{IFAC-PapersOnLine}, vol.~52, no.~20, 2019, pp. 217--222.

\bibitem{Li2006}
W.~Li and C.~G. Cassandras, ``{A Cooperative Receding Horizon Controller for
  Multi-Vehicle Uncertain Environments},'' \emph{IEEE Trans. on Automatic
  Control}, vol.~51, no.~2, pp. 242--257, 2006.

\bibitem{Chen2020}
\BIBentryALTinterwordspacing
R.~Chen and C.~G. Cassandras, ``{Optimal Assignments in Mobility-on-Demand
  Systems Using Event-Driven Receding Horizon Control},'' \emph{IEEE Trans. on
  Intelligent Transportation Systems}, pp. 1--15, 2020. [Online]. Available:
  \url{https://doi.org/10.1109/TITS.2020.3030218}
\BIBentrySTDinterwordspacing

\bibitem{Yu2016}
J.~Yu, M.~Schwager, and D.~Rus, ``{Correlated Orienteering Problem and its
  Application to Persistent Monitoring Tasks},'' \emph{IEEE Trans. on
  Robotics}, vol.~32, no.~5, pp. 1106--1118, 2016.

\bibitem{Pakdaman2009}
M.~Pakdaman and M.~M. Sanaatiyan, ``{Design and Implementation of Line Follower
  Robot},'' in \emph{Proc. of Intl. Conf. on Computer and Electrical
  Engineering}, vol.~2, 2009, pp. 585--590.

\bibitem{Kim2020}
T.~Kim, C.~Lee, and H.~Shim, ``{Completely Decentralized Design of Distributed
  Observer for Linear Systems},'' \emph{IEEE Trans. on Automatic Control},
  vol.~65, no.~11, pp. 4664--4678, 2020.

\bibitem{bryson1975}
A.~E. Bryson, Y.~C. Ho, Y.~C. Ho, and D.~P. Cantwell, \emph{{Applied Optimal
  Control: Optimization, Estimation, and Control}}.\hskip 1em plus 0.5em minus
  0.4em\relax Hemisphere Publishing Corporation, 1975.

\end{thebibliography}

\end{document}